\documentclass[10pt,a4paper]{llncs}
\usepackage{cite}
\usepackage{times}
\usepackage[linesnumbered, ruled, noend]{algorithm2e}
\usepackage[noend]{algpseudocode}
\usepackage{multirow}
\usepackage{tabularx}
\usepackage{subfigure}
\usepackage{booktabs}
\usepackage{dcolumn}
\usepackage{url}
\usepackage{amssymb}
\usepackage{amsmath}
\usepackage{tikz}

\usepackage{thmtools, thm-restate}

\usepackage{graphicx}
\usepackage{gastex}
\usepackage{thm-restate}
\usepackage{colortbl}
\usepackage{fancyvrb}
\usepackage{txfonts}
\usepackage{wrapfig}
\usepackage{pgfplots}

\newcolumntype{d}{D{,}{,}{3.3}}
\usetikzlibrary{backgrounds}
\usetikzlibrary{arrows,automata}
\usetikzlibrary{positioning}
\usetikzlibrary{calc}
\usetikzlibrary{trees}
\usepgfplotslibrary{patchplots}

\SetKwRepeat{Do}{do}{while}%

\newcommand{\trans}{\delta}
\newcommand{\alphabet}{\Sigma}
\newcommand{\emptyword}{\epsilon}
\newcommand{\finwords}{\alphabet^*}
\newcommand{\infwords}{\alphabet^\omega}
\newcommand{\poswords}{\alphabet^+}

\newcommand{\machine}{M}
\newcommand{\buechiU}{\overline{B}}
\newcommand{\buechiL}{\underline{B}}
\newcommand{\proDFA}{A}
\newcommand{\aut}{A}
\newcommand{\states}{Q}
\newcommand{\initState}{q_0}
\newcommand{\acc}{F}
\newcommand{\lang}[1]{L(#1)}
\newcommand{\upword}[1]{UP(#1)}
\newcommand{\quotient}[1]{\finwords/_{#1}}
\newcommand{\class}[1]{ [#1] }
\newcommand{\func}[1]{\mathbf{#1}}

\newcommand{\singleEq}{\backsim}
\newcommand{\doubleEq}{\approx}

\newcommand{\true}{\mathrm{T}}
\newcommand{\false}{\mathrm{F}}
\newcommand{\A}{\mathrm{A}}
\newcommand{\B}{\mathrm{B}}
\newcommand{\C}{\mathrm{C}}

\newcommand{\hide}[1]{}
\newcommand{\learnT}{\mathcal{T}}

\renewcommand{\qed}{\hfill\blacksquare}

\newcommand{\omegaRegLang}{L}
\newcommand{\canoEq}{\singleEq_{\omegaRegLang}}
\newcommand{\proEq}{\doubleEq^{u}}
\newcommand{\periodicEq}{\doubleEq^{u}_{P}}
\newcommand{\syntacticEq}{\doubleEq^{u}_{S}}
\newcommand{\recurrentEq}{\doubleEq^{u}_{R}}

\newcommand{\fdfas}{\mathcal{F}}
\newcommand{\deqWithu}[3]{#2\doubleEq^{u}_{#1} #3}

\newcommand{\eqWith}[3]{#2\singleEq_{#1} #3}

\newcommand{\stateWord}[1]{\tilde{#1}}
\newcommand{\FDFA}{\textsc{FDFA}}

\newcommand{\size}[1]{|#1|}
\newcommand{\autdollar}{\mathcal{D}}
\newcommand{\autdollarfdfa}{\mathcal{D}_1}
\newcommand{\autdollarfdfaneq}{\mathcal{D}_2}
\newcommand{\preOfeq}{\unlhd}
\newcommand{\preOfneq}{\lhd}

\newcommand{\roll}{\textsf{ROLL}}

\newcommand{\Lperiodic}{L^{\textsf{Periodic}}}
\newcommand{\Lsyntactic}{L^{\textsf{Syntactic}}}
\newcommand{\Lrecurrent}{L^{\textsf{Recurrent}}}
\newcommand{\Ldollar}{L^\$}

\title{A Novel Learning Algorithm for B\"uchi Automata\\ based on Family of DFAs and Classification Trees}
\author{
Yong Li\inst{1,2}, Yu-Fang Chen\inst{3}, Lijun Zhang\inst{1,2}, Depeng Liu\inst{1,2}
}
\institute{
	State Key Laboratory of Computer Science,
    Institute of Software, CAS
\and
	University of Chinese Academy of Sciences
\and
    Institute of Information Science, Academia Sinica
}
\begin{document}

\maketitle

\begin{abstract}
In this paper, we propose a novel algorithm to learn a B\"uchi automaton from a teacher who knows an $\omega$-regular language. The algorithm is based on learning a formalism named \emph{family of DFAs} (FDFAs) recently proposed by Angluin and Fisman~\cite{Angluin2014}. The main catch is that we use a \emph{classification tree} structure instead of the standard \emph{observation table} structure.
The worst case storage space required by our algorithm is quadratically better than the table-based algorithm proposed in~\cite{Angluin2014}.
We implement the first publicly available library {\roll} (Regular Omega Language Learning
), which consists of all $\omega$-regular learning algorithms available in the literature and the new algorithms proposed in this paper.
Experimental results show that our tree-based algorithms have the best performance among others regarding the number of solved learning tasks.
\end{abstract}

\section{Introduction}
Since the last decade, learning-based automata inference techniques~\cite{Bollig2009ALN,Kearns1994ICL,Angluin1987LS,Rivest1989IFA} have received significant attention from the community of formal system analysis. In general, the primary applications of automata learning in the community can be categorized into two: \emph{improving efficiency and scalability of verification}~\cite{cobleigh2003learning,chaki2005automated,chen2009learning,grumberg2016learning,lin2014learning,alur2005symbolic,feng2011automated,HeGWZ15} and \emph{synthesizing abstract system model for further analysis}~\cite{peled2002black,hagerer2002model,WangWHC11,alur2005synthesis,howar2013hybrid,giannakopoulou2012symbolic,sun2015tlv,aarts2015generating,ChapmanCKKST15,Chen2016PLV}.

The former usually is based on the so called \emph{assume-guarantee} compositional verification approach, which divides a verification task into several subtasks via a composition rule. Learning algorithms are applied to construct environmental assumptions of components in the rule automatically. For the latter, automata learning has been used to automatically generate interface model of computer programs~\cite{alur2005synthesis,howar2013hybrid,giannakopoulou2012symbolic,xiao2013tzuyu,sun2015tlv}, a model of system error traces for diagnosis purpose~\cite{ChapmanCKKST15}, behavior model of programs for statistical program analysis\cite{Chen2016PLV}, and model-based testing and verification~\cite{peled2002black,hagerer2002model,WangWHC11}.

Besides the classical finite automata learning algorithms, people also apply and develop learning algorithm for richer models for the above two applications. For example, learning algorithms for register automata~\cite{HowarSJC12,IsbernerHS14} have been developed and applied to synthesis system and program interface models. Learning algorithm for timed automata has been developed for automated compositional verification for timed systems~\cite{lin2014learning}. However, all the results mentioned above are for checking \emph{safety properties} or synthesizing \emph{finite behavior models} of systems/programs.
B\"uchi automaton is the standard model for describing liveness properties of distributed systems~\cite{alpern1987recognizing}. The model has been applied in automata theoretical model checking~\cite{vardi1986automata} to describe the property to be verified. It is also often used in the synthesis of reactive systems. Moreover, B\"uchi automata have been used as a means to prove program termination~\cite{lee2001size}. However, unlike the case for finite automata learning, learning algorithms for B\"uchi automata are very rarely used in our community. We believe this is a potentially fertile area for further investigation.

The first learning algorithm for the full-class of $\omega$-regular languages represented as B\"uchi automata was described in~\cite{Farzan2008}, based on the $L^*$ algorithm~\cite{Angluin1987LS} and the result of~\cite{Calbrix1993}. Recently, Angluin and Fisman propose a new learning algorithm for $\omega$-regular languages~\cite{Angluin2014} using a formalism called a \emph{family of DFAs} (FDFAs), based on the results of~\cite{Maler1993}. The main problem of applying their algorithm in verification and synthesis is that their algorithm requires a teacher for FDFAs. In this paper, we show that their algorithm can be adapted to support B\"uchi automata teachers.

We propose a novel $\omega$-regular learning algorithm based on FDFAs and a \emph{classification tree} structure (inspired by the tree-based $L^*$ algorithm in~\cite{Kearns1994ICL}).
The worst case storage space required by our algorithm is quadratically better than the table-based algorithm proposed in~\cite{Angluin2014}.
Experimental results show that our tree-based algorithms have the best performance among others regarding the number of solved learning tasks.

For regular language learning, there are robust and publicly available libraries, e.g.,  \textsf{libalf}\cite{Bollig2010} and \textsf{LearnLib}\cite{isberner2015open}. A similar library is still lacking for B\"uchi automata learning.
We implement the first publicly available B\"uchi automata learning library, named {\roll} (Regular Omega Language Learning, \url{http://iscasmc.ios.ac.cn/roll}), which includes all B\"uchi automata learning algorithms of the full class of $\omega$-regular languages available in the literature and the ones proposed in this paper. We compare the performance of those algorithms using a benchmark consists of
295 B\"uchi automata corresponding to all 295 LTL specifications available in B\"uchiStore~\cite{Tsay2011BSO}.

To summarize, our contribution includes the following. (1) Adapting the algorithm of~\cite{Angluin2014} to support B\"uchi automata teachers. (2) A novel learning algorithm for $\omega$-regular language based on FDFAs and classification trees. (3) The publicly available library {\roll} that includes all B\"uchi automata learning algorithms can be found in the literature. (4) A comprehensive empirical evaluation of B\"uchi automata learning algorithms.

\section{Preliminaries}\label{sec:preliminaries}
Let $A$ and $B$ be two sets. We use $A \oplus B$ to denote their \emph{symmetric difference}, i.e., the set $(A\setminus B) \cup (B\setminus A)$.
Let $\alphabet$ be a finite set called \emph{alphabet}. We use $\emptyword$ to represent an empty word.
The set of all finite words is denoted by $\finwords$, and the set of all infinite words, called $\omega$-words, is denoted by $\infwords$. Moreover, we also denote by $\poswords$ the set $\finwords\setminus\{\emptyword\}$.
We use $\size{u}$ to denote the length of the finite word $u$.
We use $[i\cdots j]$ to denote the set $\{i, i+1, \cdots, j\}$. We denote by $w[i]$ the $i$-th letter of a word $w$.
We use $w[i..k]$ to denote the subword of $w$ starting at the $i$-th letter and ending at the $k$-th letter, inclusive, when $i\leq k$ and the empty word $\emptyword$ when $i > k$.
A \emph{language} is a subset of $\finwords$ and an \emph{$\omega$-language} is a subset of $\infwords$.
Words of the form $uv^\omega$ are called \emph{ultimately periodic} words.
We use a pair of finite words $(u,v)$ to denote the ultimately periodic word $w= uv^\omega$. We also call $(u, v)$ a \emph{decomposition} of $w$.
For an $\omega$-language $\omegaRegLang$, let $UP(\omegaRegLang) = \{uv^\omega \mid u \in \Sigma^*,v \in \Sigma^+,uv^\omega \in \omegaRegLang \}$, i.e., all ultimately periodic words in $\omegaRegLang$.


A \emph{finite automaton} (FA) is a tuple $\aut = (\alphabet, \states, \initState,  \acc, \trans)$
consisting of a finite alphabet $\alphabet$,
a finite set $\states$ of states, an initial state $\initState$, a set $\acc\subseteq \states$ of accepting states, and a transition relation $\trans \subseteq \states\times\alphabet\times \states$.
For convenience, we also use $\trans(q, a)$ to denote the set $\{q'\mid (q, a, q') \in \trans \}$.
A \emph{run} of an FA on a finite word $v = a_1 a_2 a_3 \cdots a_n$ is a sequence of states $\initState, q_1, \cdots, q_n$
such that $(q_i, a_{i+1}, q_{i+1}) \in \trans$. The run $v$ is \emph{accepting} if $q_n\in \acc$.
A word $u$ is accepting if it has an accepting run. The
language of $A$, denoted by $\lang{\aut}$, is the set $ \{u\in\finwords \mid u\ \text{is accepted by}\ \aut\}$.
Given two FAs $A$ and $B$, one can construct a product FA $A\times B$ recognizing $L(A) \cap L(B)$ using a standard product construction.

A \emph{deterministic finite automaton} (DFA) is an FA such that $\trans(q, a)$ is a singleton for any $q\in\states$ and $a\in\alphabet$.
For DFA, we write $\trans(q, a) = q'$ instead of $\trans(q, a) = \{q'\}$.
The transition can be lifted to words by defining $\trans(q, \emptyword) = q$ and
$\trans(q,av)=\trans(\trans(q,a), v)$ for $q\in \states, a\in\alphabet$ and $v\in\finwords$.
We also use $\aut(v)$ as
a shorthand for $\trans(\initState, v)$. 

A \emph{B\"uchi automaton} (BA) has the same structure as an FA, except that it accepts only infinite words.
A run of an infinite word in a BA is an infinite sequence of states defined similarly to the case of a finite word in an FA.
An infinite word $w$ is accepted by a BA
iff it has a run visiting at least one accepting state infinitely often.
The language defined by a BA $\aut$, denoted by $\lang{\aut}$, is the set $\{w\in\infwords \mid w\ \text{is accepted by}\ \aut\}$.
An $\omega$-language $\omegaRegLang \subseteq \Sigma^\omega$ is $\omega$-regular iff there exists a BA $\aut$ such that $\omegaRegLang = \lang{\aut}$.

\begin{theorem}[Ultimately Periodic Words of $\omega$-Regular Languages~\cite{buchi1966symposium}]
Let $\omegaRegLang$, $\omegaRegLang'$ be two $\omega$-regular languages. Then $\omegaRegLang=\omegaRegLang'$ if and only if $UP(\omegaRegLang)=UP(\omegaRegLang')$.
\end{theorem}

\begin{definition}[Family of DFAs ($\FDFA$) \cite{Angluin2014}]\label{def:fdfa} A family of DFAs $\fdfas=(\machine, \{\proDFA^q\})$
	over an alphabet $\alphabet$ consists of a leading automaton $\machine=(\alphabet, \states, \initState, \trans)$
	and progress DFAs $\proDFA^q=(\alphabet, \states_q, s_q, \trans_q, \acc_q)$ for each $q\in\states$.
\end{definition}

Notice that the leading automaton $\machine$ is a DFA without accepting states.
Each $\FDFA$ $\fdfas$ characterizes a set of ultimately periodic words $UP(\fdfas)$.
Formally, an ultimately periodic word $w$ is in $UP(\fdfas)$ iff it has a decomposition $(u, v)$ \emph{accepted} by $\fdfas$. A decomposition $(u, v)$ is accepted by $\fdfas$ iff $\machine(uv) = \machine(u)$ and  $v \in \lang{\proDFA^{\machine(u)}}$.
An example of an $\FDFA$ $\fdfas$ is depicted in Fig.~\ref{fig:fdfa-example}.
The leading automaton $\machine$ has only one state $\emptyword$.
The progress automaton of $\emptyword$ is $\proDFA^{\emptyword}$.
The word $(ba)^\omega$ is in $UP(\fdfas)$ because it has a decomposition $(ba,ba)$ such that $\machine(ba\cdot ba) = \machine(ba)$ and $ba\in \lang{\proDFA^{\machine(ba)}}=\lang{\proDFA^{\emptyword}}$.
It is easy to see that the decomposition $(bab, ab)$ is not accepted by $\fdfas$
since $ab\not\in \lang{\proDFA^{\machine(bab)}}=\lang{\proDFA^{\emptyword}}$.

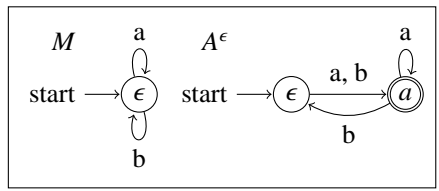
\begin{wrapfigure}{l}{0.5\textwidth}
	\vspace{-0.6cm}
	\centering
	\begin{tikzpicture}[shorten >=0.5pt,node distance=1.5cm,on grid,auto,framed]
	
	\begin{scope}
	\node[initial,state, inner sep=2.5pt,minimum size=2.5pt] (q0)      {$\emptyword$};
	\node[] at ($(q0) + (-1, 0.7)$) {$\machine$};
	
	\path[->]
	(q0) edge [loop above] node {a} (q0)
	edge [loop below] node {b} (q0);
	\end{scope}
	
	\begin{scope}[xshift=2cm]
	\node[initial,state, inner sep=2.5pt,minimum size=2.5pt] (q0)      {$\emptyword$};
	\node[state, accepting, inner sep=2pt,minimum size=2pt]         (q1) [right =of q0]  {$a$};
	\node[] at ($(q0) + (-1, 0.7)$) {$\proDFA^{\emptyword}$};
	
	\path[->] (q0)  edge node {a, b} (q1)
	(q1)  edge [loop above]  node {a} (q1)
	edge [bend left]  node {b} (q0);
	\end{scope}
	\end{tikzpicture}
	\vspace{-0.3cm}
	\caption{An example of an $\FDFA$}\label{fig:fdfa-example}
	\vspace{-0.7cm}
\end{wrapfigure}

For any $\omega$-regular language $\omegaRegLang$, there exists an $\FDFA$ $\fdfas$ such that  $UP(\omegaRegLang) = UP(\fdfas)$~\cite{Angluin2014}. We show in Sec.~\ref{sec:buechi-builder} that it is not the case for the reverse direction.
More precisely,  in~\cite{Angluin2014}, three kinds of $\FDFA$s are suggested as the canonical representation of $\omega$-regular languages, namely \emph{periodic $\FDFA$}, \emph{syntactic $\FDFA$} and \emph{recurrent $\FDFA$}. Their formal definitions are given in terms of \emph{right congruence}.

An equivalence relation $\singleEq$ on $\finwords$ is a right congruence if $x\singleEq y$ implies $xv\singleEq yv$ for every $x, y, v\in\finwords$. 
The index of $\singleEq$, denoted by $\size{$$\singleEq$$}$, is the number of equivalence classes of $\singleEq$.
We use $\quotient{\singleEq}$ to denote the equivalence classes of the right congruence $\singleEq$.
A \emph{finite right congruence} is a right congruence with a finite index.
For a word $v\in\finwords$, we use the notation $\class{v}_{\singleEq}$ to represent the class of $\singleEq$ in which $v$ resides and ignore the subscript $\singleEq$ when the context is clear.
The right congruence $\canoEq$ of a given $\omega$-regular language $\omegaRegLang$ is defined such that $x\canoEq y$ iff $\forall w\in\infwords.xw\in\omegaRegLang \Longleftrightarrow yw\in\omegaRegLang$.
The index of $\canoEq$ is finite because it is not larger than the number of states in a deterministic Muller automaton recognizing $\omegaRegLang$~\cite{Maler1993}.
\begin{definition}[Canonical $\FDFA$~\cite{Angluin2014}]\label{def:cano-fdfas}
	Given an $\omega$-regular language $\omegaRegLang$, a periodic (respectively, syntactic and recurrent) $\FDFA$ $\fdfas = (\machine, \{\proDFA^q\})$ of $L$ is defined as follows.\\
	The leading automaton $M$ is the tuple $(\alphabet, \quotient{\canoEq}, \class{\emptyword}_{\canoEq}, \trans)$, where $\trans(\class{u}_{\canoEq}, a) = \class{ua}_{\canoEq}$ for all $u\in\finwords$ and $a \in \alphabet$.
	
	We define the right congruences $\proEq_P, \proEq_S$, and $\proEq_R$ for progress automata $\proDFA^u$
	of periodic, syntactic, and recurrent $\FDFA$ respectively as follows:
	\[
	\begin{array}{lll}
	\deqWithu{P}{x}{y} \text{ iff } & &\forall v\in\finwords, u(xv)^\omega\in\omegaRegLang\Longleftrightarrow u(yv)^\omega\in\omegaRegLang, \\
	\deqWithu{S}{x}{y} \text{ iff } \text{ } \eqWith{L}{ux}{uy}& \text{and} &
	\forall v\in\finwords, \eqWith{L}{uxv}{u} \Longrightarrow (u(xv)^\omega\in\omegaRegLang\Longleftrightarrow u(yv)^\omega\in\omegaRegLang), \text{ and} \\
	\deqWithu{R}{x}{y} \text{ iff } & &
	\forall v\in\finwords,
	\eqWith{L}{uxv}{u}\land u(xv)^\omega\in\omegaRegLang \Longleftrightarrow\eqWith{L}{uyv}{u}\land u(yv)^\omega\in\omegaRegLang.
	\end{array}
	\]
	The progress automaton $\proDFA^u$ is the tuple $(\alphabet, \quotient{\proEq_K}, \class{\emptyword}_{\proEq_K}, \trans_K, \acc_K)$, where $\trans_K(\class{u}_{\proEq_K}, a) = \class{ua}_{\proEq_K}$ for all $u\in\finwords$ and $a \in \alphabet$.
	The accepting states $\acc_K$ is the set of equivalence classes $\class{v}_{\proEq_K}$ for which
	$\eqWith{L}{uv}{u}$ and $uv^\omega\in\omegaRegLang$ when $K \in \{S, R\}$ and the set of equivalence classes $\class{v}_{\proEq_K}$ for which $uv^\omega\in\omegaRegLang$ when $K \in \{P\}$.
	
\end{definition}
In this paper, by an abuse of notation, we use a finite word $u$ to
denote the state in a DFA in which the equivalence class $\class{u}$ resides.


\begin{lemma}[\hspace*{-0.12cm}\cite{Angluin2014}]\label{lem:language-eq-fdfa}
	Let $\fdfas$ be a periodic (syntactic, recurrent) $\FDFA$ of an $\omega$-regular language $\omegaRegLang$. Then $\upword{\fdfas}=\upword{\omegaRegLang}$.
\end{lemma}
\begin{lemma}[\hspace*{-0.1cm}\cite{AngluinBF16}]\label{lem:language-fdfa-to-ba}
	Let $\fdfas$ be a periodic (syntactic, recurrent) $\FDFA$ of an $\omega$-regular language $\omegaRegLang$. One can construct a BA recognizing $\omegaRegLang$ from $\fdfas$.
\end{lemma}

\section{B\"uchi Automata Learning Framework based on $\FDFA$}\label{sec:learning-framework}

\begin{figure}[h]

\scalebox{0.9}{
\begin{tikzpicture}

   \hide{
		\node[draw=black!90, thick, inner sep=0, anchor=north,rounded corners, fill opacity=0.85,text width=10.2cm,fill=gray!20,minimum height=7.8cm] (ba_learner) at (3.4,0.5) {};
		\node[] at (3.4,0.3) {BA/LBA learning framework};		
	}
		\node[draw=purple!50, anchor=north,rounded corners,  text width=13cm,fill=purple!20,minimum height=1.5cm] (Mem) at (4.5,-0.8) {};
		\node[rotate=270, align = center] at (-1.8,-1.6) {Member};		

		\node[draw=purple!50, anchor=north,rounded corners,  text width=13cm,fill=purple!20,minimum height=4.1cm] (Mem) at (4.5,-2.8) {};
		\node[rotate=270, align = center] at (-1.8,-4.8) {Equivalence};

		\node[draw=blue!50, inner sep=0, anchor=north,rounded corners, fill opacity=0.85,text width=3.2cm,fill=blue!20,minimum height=7.2cm] (learner) at (0,0) {};
		\node[] at (0,-0.3) {FDFA learner};
		
		\node[draw=blue!50, inner sep=0, anchor=north,rounded corners, fill opacity=0.85, text width=6.7cm,fill=blue!20,minimum height=7.2cm] (teacher) at (7.5,0) {};
		\node[] at (7.5,-0.3) {FDFA teacher};

		\node[draw=orange!50, inner sep=0, anchor=north,rounded corners, text width=0.6cm,fill=orange!10,minimum height=6.5cm] (BA_teacher) at (10.3,-0.25) {};
		\node[rotate=270, align = center] at (10.3,-3.5) {BA teacher};

		\node[draw=yellow!50, inner sep=0, anchor=north,rounded corners, text width=2.8cm,fill=yellow!10,minimum height=1cm] (table_learner) at (0,-1) {};
		\node[align = center, text width=2.6cm] at (0,-1.5) {
			Table-based~\cite{Angluin2014} (Sec.\ref{sec:fdfa-learner-table})
		};
		
		\node[draw=green!50, inner sep=0, anchor=north,rounded corners, text width=2.8cm,fill=green!10,minimum height=3cm] (tree_learner) at (0,-3.2) {};
		\node[align = center] at (0,-4.5) {
			Tree-based (Sec.~\ref{sec:fdfa-learner-tree})\\
			\begin{minipage}[t]{0.222\textwidth}
			\begin{itemize}
			\item Periodic FDFA
			\item Syntactic FDFA
			\item Recurrent FDFA
			\end{itemize}
			\end{minipage}			
		};
		
		\node[draw=green!50, inner sep=0, anchor=north,rounded corners, text width=4cm,fill=green!10,minimum height=1.5cm] (FDFA2BA) at (6.5,-3.2) {};
		\node[align = left] at (7,-4) {
			FDFA $F$ to BA $B$~(Sec.~\ref{sec:buechi-builder})\\
			\begin{minipage}[t]{0.4\textwidth}
			\begin{itemize}
            \item Under-Approximation $\underline B$
			\item Over-Approximation $\overline B$
			\end{itemize}
			\end{minipage}			
	    };

		\node[draw=green!50, inner sep=0, anchor=north,rounded corners, text width=4cm,fill=green!10,minimum height=1.5cm] (normalize_CE) at (6.5,-5.25) {};
		\node[align = left] at (7,-6) {
			Analyze CE~(Sec.~\ref{sec:ce-translation})\\
			\begin{minipage}[t]{0.4\textwidth}
			\begin{itemize}
			\item Under-Approximation $\underline B$
			\item Over-Approximation $\overline B$
			\end{itemize}
			\end{minipage}			
		};
		\draw [dashed, ->] (6.5,-4.75) -- node[auto] {$F$}(6.5,-5.25) ;

		\draw [->] (1.6,-1.5) -- node[auto] {$\mathsf{Mem^{FDFA}}(u,v)$}(4.15,-1.5) ;
		\draw [->] (4.15,-1.5) -- node[auto] {$\mathsf{Mem^{BA}}(uv^\omega)$}(10,-1.5) ;
		\draw [-] (10,-1.7) --(4.15,-1.7) ;	
		\draw [->] (4.15,-1.7) -- node[auto] {yes/no}(1.6,-1.7) ;	
		\draw [->] (1.6,-3.5) -- node[auto] {$\mathsf{Equ^{FDFA}}(F)$}(4.5,-3.5) ;	
		\draw [->] (8.5,-3.5) -- node[auto] {$\mathsf{Equ^{BA}}(B)$}(10,-3.5) ;	
		\draw [-] (10,-6.5) --node[auto] {yes}(8.5,-6.5) -- (8.5,-7.5) ;	
		\draw [->] (8.5,-7.5) -- (7.5,-7.5)node[left]{Output a BA recognizing the target language} ;

		\draw [->] (10,-6) -- node[above] {no + $uv^\omega$}(8.5,-6) ;				
		\draw [->] (4.5,-6) -- node[above]  {no +$(u', v')$}(1.6,-6) ;	
				
\end{tikzpicture}
}
\caption{Overview of the learning framework based on $\FDFA$ learning.
The components in \protect\tikz \protect\node[draw=yellow!50, rounded corners, text width=0.5cm,fill=yellow!10,minimum height=0.3cm] {}; boxes are results from existing works.
The components in \protect\tikz \protect\node[draw=green!50, rounded corners, text width=0.5cm,fill=green!10,minimum height=0.3cm] {}; boxes are our new contributions. }\label{fig:overview-architecture}
\end{figure}
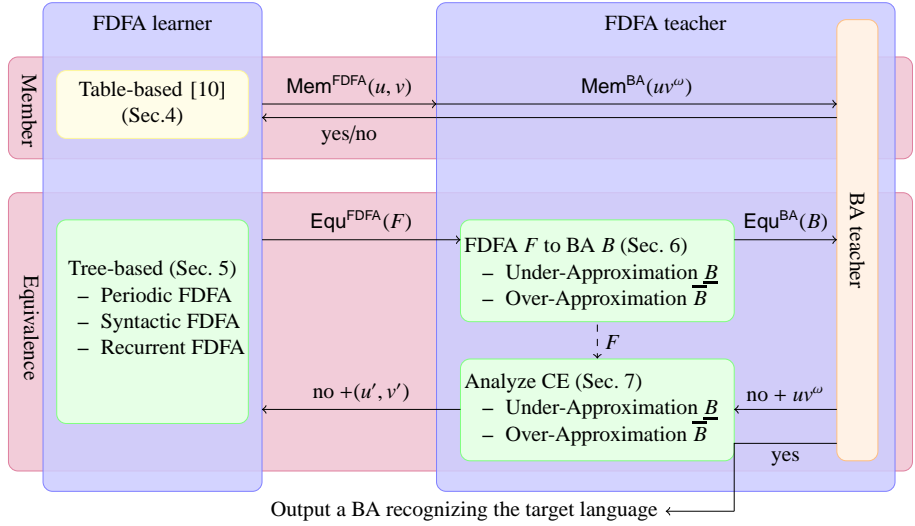

We begin with an introduction of the framework of learning BA recognizing an unknown $\omega$-regular language $\omegaRegLang$.

\subsubsection*{Overview of the framework:}
First, we assume that we already have a BA teacher who knows the unknown $\omega$-regular language $\omegaRegLang$ and
answers \emph{membership} and \emph{equivalence} queries about $\omegaRegLang$.
More precisely, a membership query $\mathsf{Mem^{BA}}(uv^\omega)$ asks if $uv^\omega \in \omegaRegLang$.
For an equivalence query $\mathsf{Equ^{BA}}(B)$, the BA teacher answers ``yes" when $L(B)=\omegaRegLang$, otherwise it returns ``no" as well as a counterexample $uv^\omega\in \omegaRegLang\oplus \lang{B}$.

The framework depicted in Fig.~\ref{fig:overview-architecture} consists of two components,
namely the \emph{$\FDFA$ learner} and the \emph{$\FDFA$ teacher}. Note that one can place any $\FDFA$ learning algorithm to the $\FDFA$ learner component. For instance, one can use the $\FDFA$ learner from \cite{Angluin2014} which employs a table to store query results, or
the $\FDFA$ learner using a classification tree proposed in this paper. The $\FDFA$ teacher
can be any teacher who can answer membership and equivalence queries about an unknown $\FDFA$.
\subsubsection*{$\FDFA$ learners:}
The $\FDFA$ learners component will be introduced in Sec.~\ref{sec:fdfa-learner-table}
and~\ref{sec:fdfa-learner-tree}.
We first briefly review the table-based $\FDFA$ learning algorithms~\cite{Angluin2014} in Sec.~\ref{sec:fdfa-learner-table}.
Our tree-based learning algorithm for canonical $\FDFA$s will be introduced in Sec.~\ref{sec:fdfa-learner-tree}.
The algorithm is inspired by the tree-based $L^*$ learning algorithm~\cite{Kearns1994ICL}.
Nevertheless, applying the tree structure to learn $\FDFA$s is not a trivial task. For example,
instead of a binary tree used in~\cite{Kearns1994ICL}, we need to use a $K$-ary tree to learn syntactic FDFAs.
The use of  $K$-ary tree complicates the procedure of refining the classification tree and automaton construction.
More details will be provided in Sec.~\ref{sec:fdfa-learner-tree}.

\subsubsection*{$\FDFA$ teacher:}
The task of the $\FDFA$ teacher is to answer queries $\mathsf{Mem^{FDFA}}(u,v)$ and $\mathsf{Equ^{FDFA}}(F)$ posed by the $\FDFA$ learner. Answering $\mathsf{Mem^{FDFA}}(u,v)$ is easy. The $\FDFA$ teacher just needs to redirect the result of $\mathsf{Mem^{BA}}(uv^\omega)$ to the $\FDFA$ learner.
Answering equivalence query $\mathsf{Equ^{FDFA}}(F)$ is more tricky.\\\ \\
\underline{From an $\FDFA$ $F$ to a BA $B$:} The $\FDFA$ teacher needs to transform an $\FDFA$ $F$ to a BA $B$ to pose an equivalence query $\mathsf{Equ^{BA}}(B)$.
In Sec.~\ref{sec:buechi-builder}, we show that, in general, it is impossible to build a BA $B$ from an $\FDFA$ $F$ such that $\upword{\lang{B}} = \upword{F}$.
Therefore in Sec.~\ref{sec:buechi-builder}, we propose two methods to approximate $\upword{F}$,
namely the \emph{under-approximation} method
and the \emph{over-approximation} method.
As the name indicates, the under-approximation (respectively, over-approximation) method constructs
a BA $B$ from $F$ such that $\upword{\lang{B}} \subseteq \upword{F}$ (respectively, $\upword{F} \subseteq \upword{\lang{B}}$).
The under-approximation method is modified from the algorithm in~\cite{Calbrix1993}.
Note that if the $\FDFA$s are the canonical representations,
the BAs built by the under-approximation method recognize the same ultimately periodic words as the $\FDFA$s, which makes it a complete method
for BA learning (Lem.~\ref{lem:language-eq-fdfa} and \ref{lem:language-fdfa-to-ba}).
As for the over-approximation method, we cannot guarantee to get a BA $B$ such that  $\upword{\lang{B}}=\upword{F}$
even if the $F$ is a canonical representation, which thus makes it an incomplete method.
However, in the worst case, the over-approximation method produces a BA whose number of states is only quadratic
in the size of the $\FDFA$.
In contrast, the number of states in the BA constructed by the under-approximation method
is cubic in the size of the $\FDFA$.\\\ \\
\underline{Counterexample analysis:} If the $\FDFA$ teacher receives ``no" and a counterexample $uv^\omega$ from the BA teacher,
the $\FDFA$ teacher has to return ``no" as well as a valid decomposition
$(u', v')$ that can be used by the $\FDFA$ learner to refine $F$.
In Sec.~\ref{sec:ce-translation}, we show how the $\FDFA$ teacher chooses
a pair $(u', v')$ from $uv^\omega$ that allows $\FDFA$ learner to refine current $\FDFA$ $F$.
As the dashed line with a label $F$ in Fig.~\ref{fig:overview-architecture} indicates,
we use the current conjectured $\FDFA$ $F$ to analyze the counterexample.
The under-approximation method
and the over-approximation method of $\FDFA$ to BA translation require different
counterexample analysis procedures. More details will be provided in Sec.~\ref{sec:ce-translation}.

Once the BA teacher answers ``yes" for the equivalence query $\mathsf{Equ^{BA}}(B)$,  the $\FDFA$ teacher
will terminate the learning procedure and outputs a BA recognizing $\omegaRegLang$.
\section{Table-based Learning Algorithm for $\FDFA$s}\label{sec:fdfa-learner-table}

In this section, we briefly introduce the table-based learner for FDFAs~\cite{Angluin2014}.
It employs a structure called \emph{observation table}~\cite{Angluin1987LS}
to organize the results obtained from queries and propose candidate FDFAs.
The table-based FDFA learner simultaneously runs several instances of DFA learners.
The DFA learners are very similar to the $L^*$ algorithm~\cite{Angluin1987LS},
except that they use different conditions to decide if two strings belong to the same state (based on Def.~\ref{def:cano-fdfas}).
More precisely, the FDFA learner uses one DFA learner $L^*_{\machine}$ for the leading automaton $\machine$, and
for each state $u$ in $\machine$, one DFA learner $L^*_{\proDFA^u}$ for each progress automaton $\proDFA^u$.
The table-based learning procedure works as follows.
The learner $L^*_{\machine}$ first closes the observation table by posing membership queries and then constructs a candidate for leading automaton $\machine$. For every state $u$ in $\machine$, the table-based algorithm runs an instance of DFA learner $L^*_{\proDFA^u}$ to find the progress automaton $\proDFA^u$.
When all DFA learners propose  candidate DFAs, the FDFA learner assembles them
to an FDFA $\fdfas = (\machine, \{\proDFA^u\})$ and then poses an equivalence query for it.
The $\FDFA$ teacher will either return \emph{``yes"} which means the learning algorithm succeeds or
return \emph{``no"} accompanying with a counterexample. Once receiving the counterexample, the table-based algorithm
will decide which DFA learner should refine its candidate DFA.
We refer interested readers to \cite{Angluin2014} for more details of the table-based algorithm.

\section{Tree-based Learning Algorithm for $\FDFA$s}\label{sec:fdfa-learner-tree}
In this section, we provide our tree-based learning algorithm for $\FDFA$s.
To that end, we first
define the classification tree structure for $\FDFA$ learning in Sec.~\ref{sec:sub:fdfa-learner-tree-data} and present the tree-based algorithm in Sec.~\ref{sec:sub:fdfa-learner-tree-algo}.

\subsection{Classification Tree Structure in Learning}\label{sec:sub:fdfa-learner-tree-data}
Here we present our classification tree structure for $\FDFA$ learning.
Compared to the classification tree defined in~\cite{Kearns1994ICL}, ours is not restricted to be a binary tree. Formally, a classification tree is a tuple $\learnT = (N, r, L_n, L_e)$ where
$N = I \cup T$ is a set of nodes consisting of the set $I$ of \emph{internal nodes} and the set $T$ of \emph{terminal nodes}, the node $r\in N$ is the root of the tree,
$L_n: N\rightarrow \finwords \cup (\finwords\times \finwords)$ labels an internal node with an \emph{experiment} and a terminal node with a \emph{state}, and
$L_e: N \times D\rightarrow N$ maps a parent node and a label to its corresponding child node, where the set of labels $D$ will be specified below.

During the learning procedure, we maintain a \emph{leading tree} $\learnT$ for the leading automaton $\machine$, and for every state $u$ in $\machine$, we keep a \emph{progress tree} $\learnT_u$ for the progress automaton $\proDFA^u$.
For every classification tree, we define a tree experiment function
$\func{TE}: \finwords \times (\finwords \cup (\finwords \times \finwords)) \rightarrow D$.
Intuitively, $\func{TE}(x, e)$ computes the entry value at row (state) $x$ and column (experiment) $e$ of an observation table in table-based learning algorithms. 
The labels of nodes in the classification tree $\learnT$ satisfy the follow invariants:
Let $t\in T$ be a terminal node labeled with a state $x=L_n(t)$.
Let $t'\in I$ be an ancestor node of $t$ labeled with an experiment $e=L_n(t')$.
Then the child of $t'$ following the label $\func{TE}(x,e)$, i.e., $L_e(t', \func{TE}(x,e))$, is either the node $t$ or an ancestor node of $t$.

\subsubsection*{Leading tree $\learnT$:}
The leading tree $\learnT$ for $\machine$ is a binary tree with labels $D=\{ \false, \true\}$.
The tree experiment function $\func{TE}(u, (x, y)) = \true$ iff $uxy^\omega \in \omegaRegLang$ (recall the definition of $\canoEq$ in Sec.~\ref{sec:preliminaries}) where $u,x,y \in\finwords$.
Intuitively, each internal node $n$ in $\learnT$ is labeled by an experiment $xy^\omega$ represented as $(x, y)$. For any word $u \in \finwords$,
 $uxy^\omega \in \omegaRegLang$ (or  $uxy^\omega \notin \omegaRegLang$) implies that  the equivalence class of $u$ lies in the $\true$-subtree (or $\false$-subtree) of $n$.


\subsubsection*{Progress tree $\learnT_u$:} The progress trees $\learnT_{u}$ and the corresponding function
$\func{TE}(x, e)$ are defined based on the right congruences $\periodicEq$, $\syntacticEq$, and $\recurrentEq$ of canonical $\FDFA$s in Def.~\ref{def:cano-fdfas}.\vspace{1mm}\\
\underline{Periodic $\FDFA$:}\label{sec:sub:fdfa-learner-tree-periodic}
The progress tree for periodic $\FDFA$ is also a binary tree labeled with $D=\{ \false, \true\}$.
The experiment function
$\func{TE}(x, e) = \true$ iff $u(xe)^\omega \in \omegaRegLang$ where $x,e\in\finwords$.\vspace{1mm}\\
\underline{Syntactic $\FDFA$:}\label{sec:sub:fdfa-learner-tree-syntactic}
The progress tree for syntactic $\FDFA$ is a $K$-ary tree with labels $D=\states\times\{\A, \B, \C\}$ where $\states$ is the set of states in the leading automaton $\machine$.
For all $x,e \in \finwords$, the experiment function $\func{TE}(x, e)= (\machine(ux),t)$,
where $t=\A$ iff $u = \machine(uxe) \wedge u(xe)^\omega \in \omegaRegLang$,
$t=\B$ iff $u = \machine(uxe) \wedge u(xe)^\omega \not\in \omegaRegLang$,
and $t=\C$ iff $u \neq  \machine(uxe)$.

For example, assuming that $\machine$ is constructed from the right congruence $\canoEq$, for any two states $x$ and $y$ such that $\func{TE}(x, e)= \func{TE}(y, e)=(z,A)$, it must be the case that $\eqWith{L}{ux}{uy}$ because $\machine(ux)=z=\machine(uy)$. Moreover, the experiment $e$ cannot distinguish $x$ and $y$ because ${uxe}\canoEq{u}\canoEq{uye}$ and both $u(xe)^\omega,u(ye)^\omega\in\omegaRegLang$.
\vspace{1mm}\\
\underline{Recurrent $\FDFA$:}\label{sec:sub:fdfa-learner-tree-recurrent}
The progress tree for recurrent $\FDFA$ is a binary tree  labeled with $D=\{ \false, \true\}$.
The function
$\func{TE}(x, e) = \true$ iff $u(xe)^\omega \in \omegaRegLang \wedge u = \machine(uxe)$ where $x,e\in\finwords$.

\subsection{Tree-based Learning Algorithm}\label{sec:sub:fdfa-learner-tree-algo}

The tree-based learning algorithm first initializes the leading tree $\learnT$ and the progress tree $\learnT_{\emptyword}$
as a tree with only one terminal node $r$ labeled by $\emptyword$.

From a classification tree $\learnT = (N, r, L_n, L_e)$, the learner constructs a candidate of a leading automaton $\machine=(\alphabet, \states, \emptyword, \trans)$ or a progress automaton $\proDFA^{u}=(\alphabet, \states, \emptyword, \trans,\acc)$ as follow.
The set of states is $\states = \{ L_n(t) \mid t \in T\}$.
Given $s\in\states$ and $a\in\alphabet$, the transition function $\trans(s, a)$ is constructed by the following procedure.
Initially the current node $n:=r$.
If $n$ is a terminal node, it returns $\trans(s, a)=L_n(n)$.
Otherwise, it picks a unique child $n'$ of $n$ with $L_e(n, \func{TE}(sa, L_n(n))) = n'$, updates the current node to $n'$,
and repeats the procedure\footnote{For syntactic $\FDFA$, it can happen that $\delta(s,a)$
goes to a ``new'' terminal node. A new state for the $\FDFA$ is identified in such a case.}.
By Def.~\ref{def:cano-fdfas}, the set of accepting states $\acc$ of a progress automaton can be identified from the structure of $\machine$ with the help of membership queries.
For periodic $\FDFA$, $\acc = \{v \mid uv^\omega\in\omegaRegLang, v\in\states\}$ and for
syntactic and recurrent $\FDFA$, $\acc = \{v \mid \eqWith{\machine}{uv}{u}, uv^\omega\in\omegaRegLang, v\in\states\}$.

Whenever the learner has constructed an $\FDFA$ $\fdfas = (\machine, \{\proDFA^u\})$,
it will pose an equivalence query for $\fdfas$. If the teacher returns ``no" and a counterexample $(u,v)$, the learner has to refine the classification tree and propose another candidate of $\FDFA$.

\begin{definition}[Counterexample for FDFA Learner]\label{def:ce-for-fdfa-learner}
Given the conjectured $\FDFA$ $\fdfas$ and the target language $\omegaRegLang$, we say that the counterexample
\begin{itemize}
	\item $(u,v)$ is \emph{positive} if $\eqWith{\machine}{uv}{u}$, $uv^\omega\in \upword{\omegaRegLang}$, and $(u,v)$ is not accepted by $\fdfas$,
	\item $(u,v)$ is \emph{negative} if $\eqWith{\machine}{uv}{u}$, $uv^\omega\not\in \upword{\omegaRegLang}$, and $(u,v)$ is accepted by $\fdfas$.
\end{itemize}
\end{definition}

We remark that in our case all counterexamples $(u, v)$ from the $\FDFA$ teacher satisfy the constraint $\eqWith{\machine}{uv}{u}$,
which corresponds to the \emph{normalized factorization} form in~\cite{Angluin2014}.

\subsubsection*{Counterexample guided refinement of $\fdfas$:}
Below we show how to refine the classification trees based on a negative counterexample $(u,v)$. The case of a positive counterexample is symmetric.
By definition, we have  $uv\sim_M u$, $uv^\omega\notin \upword{L}$ and $(u,v)$ is accepted by $\fdfas$. Let $\stateWord{u} = \machine(u)$, if $\stateWord{u}v^\omega \in \upword{L}$, the refinement of the leading tree is performed, otherwise $\stateWord{u}v^\omega \notin \upword{L}$, the  refinement of the progress tree is performed.


\subsubsection*{Refinement for the leading tree:} In the leading automaton $\machine$ of the conjectured $\FDFA$, if a state $p$ has a transition to a state $q$ via a letter $a$,
i.e, $q = \machine(p a)$, then $pa$ has been assigned to the terminal node labeled by $q$ during the construction of $\machine$. If one also finds an experiment $e$ such that $\func{TE}(q,e)\neq \func{TE}(pa,e)$, then we know that $q$ and $pa$ should not belong to the same state in a leading automaton. W.l.o.g., we assume $\func{TE}(q,e)=F$. In such a case, the leading tree can be refined by replacing the terminal node labeled with $q$ by a tree such that (i) its root is labeled by $e$, (ii) its left child is a terminal node labeled by $q$, and (iii) its right child is a terminal node labeled by $pa$.

Below we discuss how to extract the required states $p,q$ and experiment $e$.
Let $\size{u} = n$ and
$s_0 s_1 \cdots s_n$ be the run of $\machine$ over $u$. Note that $s_0 = \machine(\emptyword)=\emptyword$ and $s_n = \machine(u) =\stateWord{u}$.
From the facts that $(u,v)$ is a negative counterexample and $\stateWord{u}v^\omega \in \upword{L}$ (the condition to refine the leading tree), we have $\func{TE}(s_0, (u[1\cdots n], v))= \false\neq  \true=\func{TE}(s_n, (\emptyword, v))= \func{TE}(s_n, (u[n+1 \cdots n], v))$ because $uv^\omega\notin\upword{\omegaRegLang}$ and $\stateWord{u}v^\omega \in \upword{\omegaRegLang}$.
Recall that we have $w[j\cdots k] =\emptyword$ when $j>k$. Therefore, there must exist a smallest $j\in[1\cdots n]$ such that $\func{TE}(s_{j-1}, (u[j\cdots n], v)) \neq \func{TE}(s_{j}, (u[j+1\cdots n], v))$.
It follows that we can use the experiment $e=(u[j+1\cdots n], v)$
to distinguish $q=s_{j}$ and $pa=s_{j-1}u[j]$. 

\begin{example}\label{example1}
Consider a conjectured $\FDFA$ $\fdfas$ in Fig.~\ref{fig:fdfa-example} produced during the process of learning $\omegaRegLang = a^\omega + b^\omega$.
The corresponding leading tree $\learnT$ and the progress tree $\learnT_\emptyword$ are depicted on the left of Fig.~\ref{fig:refine-leading-tree}.
The dotted line is for the $\false$ label and the solid one is for the $\true$ label.
Suppose the $\FDFA$ teacher returns a negative counterexample $(ab, b)$.
The leading tree has to be refined since $\machine(ab)b^\omega=b^\omega\in \omegaRegLang$.
We find an experiment $(b, b)$ to differentiate $\emptyword$ and $a$ using the procedure above and update the leading tree $\learnT$ to $\learnT'$. The leading automaton $\machine$ constructed from $\learnT'$ is depicted on the right of Fig.~\ref{fig:refine-leading-tree}.
\end{example}
\vspace{-.7cm}
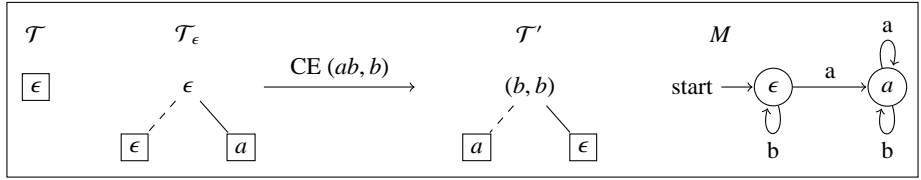
\begin{figure}
\begin{tikzpicture}[shorten >=0.5pt,node distance=1.5cm,on grid,auto,framed]
    \begin{scope}
        \node [draw] (q0) {$\emptyword$};
        \node[] at ($(q0) + (0, 0.7)$) {$\learnT$};
    \end{scope}

    \begin{scope}[xshift=2cm]
       \node[] (q01) {$\emptyword$};
       \node[ draw] (q02) at ($(q01) + (-0.7, -0.8)$) {$\emptyword$};
       \node[ draw] (q03) at ($(q01) + (0.7, -0.8)$) {$a$};
       \node[] at ($(q01) + (0, 0.7)$) {$\learnT_\emptyword$};
       \draw[dashed] (q01)  --  (q02);
       \draw (q01)  --  (q03);
       \draw[->] ($(q01) + (1, 0)$) -- node[auto] {CE $(ab, b)$} ($(q01) + (3, 0)$);
    \end{scope}

    \begin{scope}[xshift=6.5cm]
       \node[] (q01) {$(b, b)$};
       \node[ draw] (q02) at ($(q01) + (-0.7, -0.8)$) {$a$};
       \node[ draw] (q03) at ($(q01) + (0.7, -0.8)$) {$\emptyword$};
       \node[] at ($(q01) + (0, 0.7)$) {$\learnT'$};
       \draw[dashed] (q01)  --  (q02);
       \draw (q01)  --  (q03);
    \end{scope}

    \begin{scope}[xshift=9.7cm]
      \node[initial,state, inner sep=3pt,minimum size=2pt] (q0)      {$\emptyword$};
      \node[state, inner sep=3pt,minimum size=2pt]         (q1) [right =of q0]  {$a$};
      \node[] at ($(q0) + (-0.7, 0.7)$) {$\machine$};

      \path[->] (q0)  edge node {a} (q1)
                      edge [loop below] node {b} (q0)
                (q1)  edge [loop above]  node {a} (q1)
                      edge [loop below]  node {b} (q1);
    \end{scope}
\end{tikzpicture}
\vspace{-.3cm}
\caption{Refinement of the leading tree and the corresponding leading automaton}
\vspace{-1.2cm}
\label{fig:refine-leading-tree}
\end{figure}

\subsubsection*{Refinement for the progress tree:}
Here we explain the case of periodic $\FDFA$s. The other cases are similar and we leave the details in Appendix~\ref{app:ref_pt}.
Recall that $\stateWord{u}v^\omega \notin \upword{\omegaRegLang}$ and thus the algorithm refines the progress tree $\learnT_{\stateWord{u}}$.
Let $\size{v} = n$ and $h = s_0 s_1 \cdots s_n$ be the corresponding run of $\proDFA^{\stateWord{u}}$ over $v$.
Note that $s_0 = \proDFA^{\stateWord{u}}(\emptyword)=\emptyword$ and $s_n = \proDFA^{\stateWord{u}}(v)= \stateWord{v} $.
We have $\stateWord{u}(\stateWord{v})^\omega \in\upword{\omegaRegLang}$ because
$\stateWord{v}$ is an accepting state.
From the facts that $\stateWord{u}v^\omega \notin \upword{\omegaRegLang}$ and $\stateWord{u}(\stateWord{v})^\omega \in\upword{\omegaRegLang}$,
we have $\func{TE}(s_0, v[1\cdots n])=\false\neq \true=\func{TE}(s_n, \emptyword)=\func{TE}(s_n, v[n+1\cdots n])$.
Therefore, there must exist a smallest $j\in[1\cdots n]$ such that $\func{TE}(s_{j-1}, v[j\cdots n]) \neq \func{TE}(s_{j}, v[j+1\cdots n])$.
It follows that we can use the experiment $e=v[j+1\cdots n]$
to distinguish $q=s_{j}$, $pa=s_{j-1}v[j]$ and refine the progress tree $\learnT_{\stateWord{u}}$.

\subsubsection*{Optimization:}
Example~\ref{example1} also illustrates the fact that the counterexample
$(ab, b)$ may not be eliminated right away after the refinement.
In this case, it is still a valid counterexample (assuming that the progress tree $\learnT_\emptyword$ remains unchanged).
Thus as an optimization in our tool, one can repeatedly use the counterexample until it is eliminated.

\section{From $\FDFA$ to B\"uchi Automata}\label{sec:buechi-builder}

\begin{wrapfigure}{r}{0.45\textwidth}
	\vspace{-0.8cm}
	\centering
\begin{tikzpicture}[shorten >=0.5pt,node distance=1.2cm,on grid,auto,framed, scale=0.8]

    \begin{scope}
       \node[initial,state, inner sep=2.5pt,minimum size=2.5pt] (q0)      {$\emptyword$};
       \node[] at ($(q0) + (-1, 0.7)$) {$\machine$};

       \path[->]
           (q0) edge [loop above] node {a} (q0)
                edge [loop below] node {b} (q0);
    \end{scope}

    \begin{scope}[xshift=2.5cm]
      \node[initial,state, inner sep=2.5pt,minimum size=2.5pt] (q0)      {$\emptyword$};
      \node[state, accepting, inner sep=2pt,minimum size=2pt]         (q1) [right =of q0]  {$a$};
      \node[state, inner sep=2pt,minimum size=2pt]         (q2) [below =of q0]  {$b$};
      \node[] at ($(q0) + (-1, 0.7)$) {$\proDFA^{\emptyword}$};

      \path[->] (q0)  edge node {a} (q1)
                      edge node {b} (q2)
                (q1)  edge [loop above]  node {b} (q1)
                      edge  node {a} (q2)
                (q2)  edge [loop left]  node {a} (q2)
                      edge [loop right] node {b} (q2);
    \end{scope}
\end{tikzpicture}
\vspace{-0.3cm}
\caption{An $\FDFA$ $\fdfas$ such that $UP(\fdfas)$ does not characterize an $\omega$-regular language}\label{fig:fdfa-example-nonregular}
\vspace{-1cm}
\end{wrapfigure}
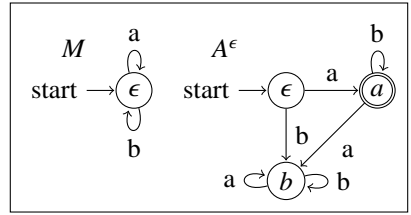
Since the FDFA teacher exploits the BA teacher for answering equivalence queries, it needs first to convert the given FDFA into a BA.
Unfortunately, with the following example, we show that in general, it is impossible to construct a \emph{precise} BA $B$ for an $\FDFA$ $\fdfas$ such that $\upword{\lang{B}} = \upword{\fdfas}$.

\begin{example}
Consider a non-canonical $\FDFA$ $\fdfas$ in Fig.~\ref{fig:fdfa-example-nonregular},
we have $\upword{\fdfas} = \bigcup_{n=0}^{\infty} \{a,b\}^* \cdot (ab^n)^\omega$.
We assume that $\upword{\fdfas}$ characterizes an $\omega$-regular language $\omegaRegLang$. It is known that the periodic $\FDFA$ recognizes exactly the $\omega$-regular language and
the index of each right congruence is finite \cite{Angluin2014}.
However, we can show that the right congruence $\doubleEq^{\emptyword}_P$ of a periodic $\FDFA$ of $\omegaRegLang$ is of infinite index.
Observe that $ab^k\not\doubleEq^{\emptyword}_P ab^j$ for any $k, j\geq 1$ and $k \neq j$, because $\emptyword\cdot(ab^k \cdot ab^k)^\omega \in\upword{\fdfas}$ and $\emptyword\cdot(ab^j \cdot ab^k)^\omega \notin\upword{\fdfas}$.
It follows that $\doubleEq^{\emptyword}_P$ is of infinite index. We conclude that $\upword{\fdfas}$ cannot characterize an $\omega$-regular language. 
\end{example}

We circumvent the above problem by proposing two BAs $\underline B, \overline B$, which under- and over-approximate
the ultimately periodic words of an $\FDFA$.
Given an $\FDFA$ $\fdfas = (\machine, \{\proDFA^u\})$ with $\machine=(\alphabet, \states, \initState, \trans)$ and $\proDFA^u=(\alphabet, \states_u, s_u, \trans_u, \acc_u)$ for all $u\in Q$, we define $\machine^s_v=(\alphabet, \states, s, \trans, \{v\})$ and $(\proDFA^{u})^s_{v}=(\alphabet, \states_u, s, \trans_u, \{v\})$, i.e., the DFA obtained from $\machine$ and $\proDFA^{u}$ by setting their initial and accepting states as $s$ and $\{v\}$, respectively.
Define $N_{(u,v)}=\{v^\omega\mid \eqWith{\machine}{uv}{u} \wedge v\in\lang{(\proDFA^{u})^{s_u}_{v}}\}$. Then $\upword{\fdfas} = \bigcup_{u\in\states, v\in\acc_u} \lang{\machine^{\initState}_u}\cdot N_{(u,v)}$.

We construct $\overline B$ and $\underline B$ by approximating the set $N_{(u,v)}$. For $\overline B$, we first define an FA $\overline P_{(u,v)}=( \alphabet, \states_{u,v}, s_{u,v}, \{f_{u,v}\}, \trans_{u,v})=\machine^u_u \times (\proDFA^{u})^{s_u}_{v}$ and let $\overline N_{(u,v)}=\lang{\overline P_{(u,v)}}^\omega$.
Then one can construct a BA $( \alphabet, \states_{u,v}\cup\{f\}, s_{u,v}, \{f\}, \trans_{u,v}\cup \trans_f)$ recognizing $\overline N_{(u,v)}$ where $f$ is a ``fresh'' state and $\trans_f=\{(f, \epsilon, s_{u,v}), (f_{u, v}, \epsilon, f)\}$.
For $\underline B$, we define an FA $\underline P_{(u,v)}=\machine^u_u\times (\proDFA^{u})^{s_u}_{v}\times (\proDFA^{u})^v_{v}$ and let $\underline N_{(u,v)}=\lang{\underline P_{(u,v)}}^\omega$. One can construct a BA recognizing $\underline N_{(u,v)}$ using a similar construction to the case of $\overline N_{(u,v)}$.
In Def.~\ref{def:fdfa-to-buechi} we show how to construct BAs $\overline B$ and $\underline B$ s.t. $\upword{\lang{\overline B}} = \bigcup_{u\in\states, v\in\acc_u} \lang{\machine^{\initState}_u}\cdot \overline N_{(u,v)}$ and $\upword{\lang{\underline B}} = \bigcup_{u\in\states, v\in\acc_u} \lang{\machine^{\initState}_u}\cdot \underline N_{(u,v)}$.

\begin{definition}\label{def:fdfa-to-buechi}
Let $\fdfas =(\machine, \{\proDFA^u\})$ be an $\FDFA$ where
$\machine = ( \alphabet, \states, \initState, \trans)$
and $\proDFA^{u} = ( \alphabet, \states_{u}, s_{u}, \acc_u, \trans_{u})$ for every $u \in \states$.
Let $( \alphabet, \states_{u,v}, s_{u,v}, \{f_{u,v}\}, \trans_{u,v})$ be a BA recognizing $\underline N_{(u,v)}$ (respectively $\overline N_{(u,v)}$). Then the BA $\underline B$ (respectively $\overline B$) is defined as the tuple
$$\left( \alphabet, \states\cup \bigcup\limits_{u\in\states, v\in \acc_u} \states_{u,v}, \initState, \bigcup\limits_{u\in\states, v\in \acc_u} \{f_{u, v} \}, \trans \cup \bigcup\limits_{u\in\states, v\in\acc_u} \trans_{u,v}\cup \bigcup\limits_{u\in\states, v\in \acc_u} \{(u, \epsilon, s_{u,v})\}\right).$$
\end{definition}

\begin{restatable}[Sizes and Languages of $\buechiL$ and $\buechiU$]{lemma}{bainclusion}
\label{lem:fdfa-to-buechi-inclusion}
Let $\fdfas$ be an FDFA and $\buechiL$, $\buechiU$ be the BAs constructed from $\fdfas$ by Def.~\ref{def:fdfa-to-buechi}.
Let $n$ and $k$ be the numbers of states in the leading automaton and the largest progress automaton of $\fdfas$.
The number of states of $\buechiL$ and $\buechiU$ are in $\mathcal{O}(n^2k^3)$ and $\mathcal{O}(n^2k^2)$, respectively.
Moreover, $\upword{\lang{\buechiL}}\subseteq \upword{\fdfas}\subseteq \upword{\lang{\buechiU}}$ and we have $\upword{\lang{\buechiL}} = \upword{\fdfas}$ when $\fdfas$ is a canonical $\FDFA$.
\end{restatable}

The properties below will be used later in analyzing counterexamples.

\begin{restatable}{lemma}{underword}\label{lem:word-fdfa-buechiL-preserve}
Given an $\FDFA$ $\fdfas = (\machine, \{\proDFA^u\})$, and $\buechiL$ the BA constructed from $\fdfas$
by Def.~\ref{def:fdfa-to-buechi}. If $(u, v^k)$ is accepted by $\fdfas$ for every $k\geq 1$, then
$uv^\omega \in \upword{\lang{\buechiL}}$.
\end{restatable}

\begin{restatable}{lemma}{overword}\label{lem:lang-buechiU-period}
Given an $\omega$-word $w \in \upword{\lang{\overline B}}$,
there exists a decomposition $(u, v)$ of $w$ and $n\geq 1$ such that
$v = v_1\cdot v_2\cdots v_n$ and for all $i \in [1\cdots n]$, $v_i\in\lang{\proDFA^{\machine(u)}}$ and $\eqWith{\machine}{uv_i}{u}$.
\end{restatable}

Fig.~\ref{fig:example-fdfa-to-buechi} depicts the BAs $\buechiU$ and $\buechiL$ constructed from the $\FDFA$ $\fdfas$ in Fig.~\ref{fig:fdfa-example}.
In the example, we can see that the $b^\omega\in\upword{\fdfas}$ while $b^\omega\notin\upword{\lang{\buechiL}}$.

\begin{figure}
\centering
\begin{tikzpicture}[shorten >=1pt,node distance=1.5cm,on grid,auto,framed]

    \begin{scope}
       \node[initial,state, inner sep=3pt,minimum size=0pt] (q0)      {$q_0$};
       \node[state, inner sep=3pt,minimum size=0pt]  (q1) [right =of q0]     {$q_1$};

       \node[state, inner sep=3pt,minimum size=0pt] (q2) [right =of q1] {$q_2$};
       \node[state, accepting, inner sep=1.5pt,minimum size=0pt] (e3) [below right =of q1] {$q'_2$};
       \node[] at ($(q0) + (-1, 0.7)$) {$\buechiU$};

       \path[->]
           (q0) edge [loop above] node {a} (q0)
                edge [loop below] node {b} (q0)
                edge node {$\epsilon$}   (q1)
           (q1) edge node {a, b} (q2)
           (q2) edge [loop above] node {a} (q2)
                edge [bend left] node {b} (q1)
                edge node {$\epsilon$} (e3)
           (e3) edge [bend left] node {$\epsilon$} (q1)
           ;
    \end{scope}
    \begin{scope}[xshift=5.7cm]
           \node[initial,state, inner sep=3pt,minimum size=0pt] (q0)      {$q_0$};
       \node[state, inner sep=3pt,minimum size=0pt]  (q1) [right =of q0]     {$q_1$};

       \node[state, inner sep=3pt,minimum size=0pt] (q2) [right =of q1] {$q_2$};
       \node[state, inner sep=3pt,minimum size=0pt] (q3) [below =of q1] {$q_3$};
       \node[state, accepting, inner sep=1.5pt,minimum size=0pt] (e3) [right =of q2] {$q'_2$};

       \node[state, inner sep=3pt,minimum size=0pt] (q4) [below =of q2] {$q_4$};
       \node[] at ($(q0) + (-1, 0.7)$) {$\buechiL$};
       \path[->]
           (q0) edge [loop above] node {a} (q0)
                edge [loop below] node {b} (q0)
                edge node {$\epsilon$}   (q1)
           (q1) edge node {a} (q2)
                edge node {b} (q3)
           (q2) edge [loop above] node {a} (q2)
                edge [bend right] node {b} (q4)
                edge node {$\epsilon$} (e3)
           (q3) edge node {a} (q2)
                edge [bend left] node {b} (q1)
           (q4) edge [bend right] node[above] {a,b} (q2)
           (e3) edge [bend right = 70] node[above] {$\epsilon$} (q1)

           ;
    \end{scope}
\end{tikzpicture}
\caption{NBA $\buechiU$ and $\buechiL$ for $\fdfas$ in Fig.~\ref{fig:fdfa-example}}\label{fig:example-fdfa-to-buechi}
\end{figure}

\section{Counterexample Analysis for FDFA Teacher}\label{sec:ce-translation}
During the learning procedure, if we failed the equivalence query for the BA $B$, the BA teacher will return a counterexample $uv^\omega$ to the FDFA teacher.

\begin{definition}[Counterexample for the FDFA Teacher]\label{def:ce-for-fdfa-teacher}
Given the conjectured BA $B\in\{\buechiL,\buechiU\}$, the target language $\omegaRegLang$, we say that
\begin{itemize}
\item
$uv^\omega$ is a \emph{positive counterexample} if $uv^\omega\in \upword{\omegaRegLang}$ and $uv^\omega\not\in \upword{\lang{B}}$,
\item
$uv^\omega$ is a \emph{negative counterexample} if $uv^\omega\not\in \upword{\omegaRegLang}$ and $uv^\omega\in \upword{\lang{B}}$.
\end{itemize}
\end{definition}

Obviously, the above is different to the counterexample for the FDFA learner in Def.~\ref{def:ce-for-fdfa-learner}.
Below we illustrate the necessity of the counterexample analysis by an example.

\begin{example}
Again, consider the conjectured $\FDFA$ $\fdfas$
depicted in Fig.~\ref{fig:fdfa-example} for $\omegaRegLang = a^\omega + b^\omega$.
Suppose the BA teacher returns a negative counterexample $(ba)^\omega$.
In order to remove $(ba)^\omega \in \upword{\fdfas}$, one has to
find a decomposition of $(ba)^\omega$ that $\fdfas$ accepts, which
is the goal of the counterexample analysis. Not all decompositions of $(ba)^\omega$ are accepted by $\fdfas$.
For instance, $(ba, ba)$ is accepted
while $(bab, ab)$ is not.

\end{example}

A positive (respectively negative) counterexample $uv^\omega$ for the $\FDFA$ teacher is \emph{spurious} if $uv^\omega \in \upword{\fdfas}$ (respectively $uv^\omega \not\in \upword{\fdfas}$).
Suppose we use the under-approximation method to construct the BA $\buechiL$ from $\fdfas$ depicted in Fig.~\ref{fig:example-fdfa-to-buechi}. The BA teacher returns a spurious positive counterexample $b^\omega$, which is in $\upword{\fdfas}$ but not in $\upword{\lang{\buechiL}}$. We show later that in such a case, one can always find a decomposition, in this example $(b, bb)$, as the counterexample for the $\FDFA$ learner.

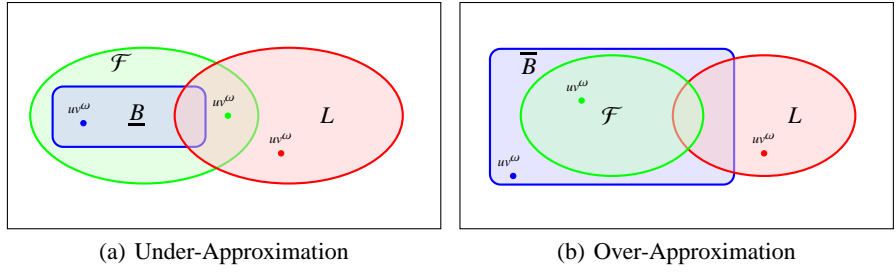
\begin{figure}
\centering
\subfigure[Under-Approximation]
{
\begin{tikzpicture}[shorten >=1pt,node distance=1.7cm,on grid,auto]
    \draw[use as bounding box] (0,1.5) rectangle (5.7,4.5);
    \draw[draw=green,thick, fill=green!20,fill opacity=0.5] (1.8,3) ellipse (1.5cm and 0.9cm);
    \node[draw=blue,thick, fill=blue!20,fill opacity=0.5, anchor=north,rounded corners, text width=1.8cm,minimum height=0.8cm] (buechiL) at (1.6,3.4) {};
    \draw[draw=red,thick,fill=red!20,fill opacity=0.5] (3.7,3) ellipse (1.5cm and 0.9cm);
    \node[] at (4.2, 3) {$\omegaRegLang$};
    \node[] at (1.7, 3) {$\buechiL$};
    \node[] at (1.5, 3.7) {$\fdfas$};
    \node[font =\tiny] at (3.6,2.7) {$uv^\omega$};
    \draw[red,fill=red] plot [mark=*, mark size=1] coordinates{(3.6,2.5)};

    \node[font =\tiny, text width=0.4cm,minimum height=0.1cm] at (1,3.1) {$uv^\omega$};
    \draw[blue,fill=blue] plot [mark=*, mark size=1] coordinates{(1,2.9)};

    \node[font =\tiny, text width=0.4cm,minimum height=0.1cm] at (2.9,3.2) {$uv^\omega$};
    \draw[green,fill=green] plot [mark=*, mark size=1] coordinates{(2.9,3)};

\end{tikzpicture}
}
\subfigure[Over-Approximation]
{
\begin{tikzpicture}[shorten >=1pt,node distance=1.7cm,on grid,auto]
    \draw[use as bounding box] (0,1.5) rectangle (5.7,4.5);
    \node[draw=blue!120,thick, fill=blue!20, fill opacity=0.5, anchor=north,rounded corners,  text width=3cm,minimum height=1.8cm] (buechiU) at (2,3.9) {};
    \draw[draw=red, thick, fill=red!20,fill opacity=0.5] (4,3) ellipse (1.2cm and 0.8cm);
    \draw[draw=green,thick, fill=green!20, fill opacity=0.5] (2,3) ellipse (1.2cm and 0.8cm);
    \node[] at (4.4, 3) {$\omegaRegLang$};
    \node[] at (2, 3) {$\fdfas$};
    \node[] at (0.9, 3.7) {$\buechiU$};

    \node[font =\tiny] at (4,2.7) {$uv^\omega$};
    \draw[red,fill=red] plot [mark=*, mark size=1] coordinates{(4,2.5)};

    \node[font =\tiny, text width=0.4cm,minimum height=0.1cm] at (1.6,3.4) {$uv^\omega$};
    \draw[green,fill=green] plot [mark=*, mark size=1] coordinates{(1.6,3.2)};

    \node[font =\tiny, text width=0.4cm,minimum height=0.1cm] at (0.7,2.4) {$uv^\omega$};
    \draw[blue,fill=blue] plot [mark=*, mark size=1] coordinates{(0.7,2.2)};
\end{tikzpicture}
}
\caption{The Case for Counterexample Analysis}\label{fig:case-ce-analysis}
\end{figure}
Given $\FDFA$ $\fdfas = (\machine, \{\proDFA^u\})$,
in order to analyze the counterexample $uv^\omega$ , we define:
\begin{itemize}
\item
an FA $\autdollar_{u\$v}$ with $\lang{\autdollar_{u\$v}} = \{u'\$v' \mid u'\in\finwords, v'\in\poswords, uv^\omega = u'v'^\omega\}$,
\item an FA $\autdollarfdfa$ with $\lang{\autdollarfdfa} = \{ u \$ v \mid u\in\finwords, v\in\finwords
, \eqWith{\machine}{uv}{u}, v\in\lang{\proDFA^{\machine(u)}}\}$, and
\item an FA $\autdollarfdfaneq$
with $\lang{\autdollarfdfaneq} = \{ u \$ v \mid u\in\finwords, v\in\finwords
, \eqWith{\machine}{uv}{u}, v\notin\lang{\proDFA^{\machine(u)}}\}$.
\end{itemize}
Here $\$$ is a letter not in $\Sigma$.
Intuitively, $\autdollar_{u\$v}$ accepts every possible decomposition $(u', v')$ of $uv^\omega$,
$\autdollarfdfa$ recognizes every decomposition $(u', v')$ which is accepted by $\fdfas$ and
$\autdollarfdfaneq$ accepts every decomposition $(u', v')$ which is not accepted by $\fdfas$
yet $\eqWith{\machine}{u'v'}{u'}$.

Given a BA $\buechiL$ constructed by the under-approximation method to construct a BA $\buechiL$ from $\fdfas$, we have that $\upword{\lang{\buechiL}} \subseteq \upword{\fdfas}$. Fig.~\ref{fig:case-ce-analysis}(a)
depicts all possible cases of $uv^\omega\in \upword{\lang{\buechiL}}\oplus \upword{\omegaRegLang}$.
\begin{itemize}
\item[U1]: $uv^\omega \in \upword{\omegaRegLang}\land uv^\omega\notin\upword{\fdfas}$ (Point in red). The word $uv^\omega$ is a positive
counterexample, one has to find a decomposition $(u', v')$ such that $\eqWith{\machine}{u'v'}{u'}$
and $u'v'^\omega = uv^\omega$. This can be easily done by taking a word $u'\$v' \in \lang{\autdollar_{u\$v}}\cap \lang{\autdollarfdfaneq}$.
\item[U2]: $uv^\omega \notin \upword{\omegaRegLang}\land uv^\omega\in\upword{\fdfas}$ (Point in blue). The word $uv^\omega$ is a negative counterexample, one needs to find a decomposition $(u', v')$ of  $uv^\omega$ that is accepted by $\fdfas$. This can be done by taking a
word $u'\$v' \in \lang{\autdollar_{u\$v}}\cap \lang{\autdollarfdfa}$.
\item[U3]: $uv^\omega \in \upword{\omegaRegLang}\land uv^\omega\in\upword{\fdfas}$ (Point in green).
The word $uv^\omega$ is a spurious positive counterexample.
Suppose the decomposition $(u, v)$ of $uv^\omega$ is accepted by $\fdfas$, according to Lem.~\ref{lem:word-fdfa-buechiL-preserve},
there must exist some $k\geq 1$ such that $(u, v^k)$ is not accepted by $\fdfas$. Thus, we can also use the same
method in U1 to get a counterexample $(u', v')$.
\end{itemize}

We can also use the over-approximation construction to get a BA $\buechiU$ from $\fdfas$ such that $\upword{\fdfas}\subseteq\upword{\lang{\buechiU}}$,
and all possible cases for a counterexample $uv^\omega\in\upword{\lang{\buechiU}}\oplus \upword{\omegaRegLang}$ is depicted in Fig.~\ref{fig:case-ce-analysis}(b).
\begin{itemize}
\item[O1]: $uv^\omega \in \upword{\omegaRegLang}\land uv^\omega\notin\upword{\fdfas}$ (Point in red).
The word $uv^\omega$ is a positive counterexample that can be dealt with the same method for case U1.
\item[O2]: $uv^\omega \notin \upword{\omegaRegLang}\land uv^\omega\in\upword{\fdfas}$ (Point in green). The word $uv^\omega$ is a negative
counterexample that can be dealt with the same method for case U2.
\item[O3]: $uv^\omega \notin \upword{\omegaRegLang}\land uv^\omega\notin\upword{\fdfas}$ (Point in blue).
In this case, $uv^\omega$ is a spurious negative counterexample. In such a case it is possible that we cannot find a valid decomposition of $uv^\omega$
to refine $\fdfas$.
By Lem.~\ref{lem:lang-buechiU-period}, we can find a decomposition $(u', v')$ of $uv^\omega$ such that
$v' = v_1v_2\cdots v_n$, $\eqWith{\machine}{u'v_i}{u'}$, and $v_i\in\lang{\proDFA^{\machine(u')}}$ for some $n\geq 1$.
It follows that $(u',v_i)$ is accepted by $\fdfas$.
If we find some $i\in [1\cdots n]$ such that $u'v_i^\omega \notin \upword{\omegaRegLang}$, then we return $(u',v_i)$, otherwise, the algorithm aborts with an error.

\end{itemize}

Finally, we note that determining whether $uv^\omega\in\upword{\omegaRegLang}$ can be done by posing a membership query $\mathsf{Mem^{BA}}(uv^\omega)$, and checking
whether $uv^\omega\in\upword{\fdfas}$ boils down to checking the emptiness of $\lang{\autdollar_{u\$v}}\cap \lang{\autdollarfdfa}$.
The construction for $\autdollar_{u\$v}$,  $\autdollarfdfa$, and $\autdollarfdfaneq$, and the correctness proof of
counterexample analysis are given in Appendix~\ref{app:fa-construction}.

\section{Complexity}

We discuss the complexity of tree-based $\FDFA$ learning algorithms in Sec.~\ref{sec:fdfa-learner-tree}.
Let $\fdfas = (\machine, \{\proDFA^u\})$ be the corresponding periodic $\FDFA$ of the $\omega$-regular language $\omegaRegLang$,
and let $n$ be the number of states in the leading automaton $\machine$ and $k$ be the number of states in the largest
progress automaton $\proDFA^u$.
We remark that $\fdfas$ is uniquely defined for $\omegaRegLang$ and
the table-based algorithm needs the same amount of equivalence queries as the tree-based one in the worst case. Given a counterexample $(u,v)$ returned from the $\FDFA$ teacher, we define its \emph{length} as $|u|+|v|$.

\begin{restatable}[Query Complexity]{theorem}{treecomplexity}\label{thm:algo-tree-halt-correctness}
Let $(u, v)$ be the longest counterexample returned from the $\FDFA$ teacher.
The number of equivalence queries needed for
the tree-based $\FDFA$ learning algorithm to learn the periodic $\FDFA$ of $\omegaRegLang$
is in $\mathcal{O}(n + nk)$, while the number of membership queries is in
$\mathcal{O}((n + nk)\cdot (\size{u} + \size{v} + (n+k)\cdot\size{\alphabet}))$.

For the syntactic and recurrent $\FDFA$s, the number of equivalence queries needed for
the tree-based $\FDFA$ learning algorithm is in $\mathcal{O}(n + n^3k)$, while
the number of membership queries is in
$\mathcal{O}((n + n^3k)\cdot (\size{u} + \size{v} + (n+nk)\cdot\size{\alphabet}))$.
\end{restatable}
The learning of syntactic and recurrent $\FDFA$s requires more queries since once their leading automata have been modified, they need to redo the learning of all progress automata from scratch.

\begin{restatable}[Space Complexity]{theorem}{treespacecomplexity}
	For all tree-based algorithms, the space required to learn the leading automaton is in $\mathcal{O}(n)$. For learning periodic $\FDFA$, the space required for each progress automaton is in $\mathcal{O}(k)$, while for syntactic and recurrent $\FDFA$s, the space required is in $\mathcal{O}(nk)$.
	For all table-based algorithms, the space required to learn the leading automaton is in $\mathcal{O}((n+n \cdot\ \size{\alphabet}) \cdot n)$. For learning periodic $\FDFA$, the space required for each progress automaton is in $\mathcal{O}((k+k \cdot\ \size{\alphabet}) \cdot k)$, while for syntactic and recurrent $\FDFA$s, the space required is in $\mathcal{O}((nk+nk \cdot\ \size{\alphabet}) \cdot nk)$.
\end{restatable}

\begin{restatable}[Correctness and Termination]{theorem}{batermination}\label{thm:termination-buechi}
The BA learning algorithm based on the under-approximation method always terminates and returns a BA
recognizing the unknown $\omega$-regular language $\omegaRegLang$ in polynomial time.
If the BA learning algorithm based on the over-approximation method
terminates without reporting an error, it returns a BA recognizing $\omegaRegLang$.
\end{restatable}

Given a canonical $\FDFA$ $\fdfas$, the under-approximation method produces a BA $\buechiL$ such that
$\upword{\fdfas} = \upword{\lang{\buechiL}}$, thus in the worst case, $\FDFA$ learner learns a canonical $\FDFA$ and terminates. In practice, the algorithm very often finds a BA recognizing $\omegaRegLang$
before converging to a canonical $\FDFA$.

\section{Experimental results}

The $\roll$ library (\url{http://iscasmc.ios.ac.cn/roll}) is implemented in JAVA.
The DFA operations in $\roll$ are delegated to the \emph{dk.brics.automaton} package, and we use the RABIT tool~\cite{AbdullaCCHHMV10,AbdullaCCHHMV11} to check the equivalence of two BAs.
We evaluate the performance of $\roll$ using the smallest BAs corresponding to all the 295 LTL specifications available in B\"uchiStore\cite{Tsay2011BSO}, where the numbers of states in
the BAs range over 1 to 17 and transitions range over 0 to 123.
The machine we used for the experiments is a 2.5 GHz Intel Core i7-6500 with 4 GB RAM. We set the timeout period to 30 minutes.

\begin{table}[!h]
   \caption{Overall experimental results. We show the results of 285 cases where all algorithms can finish the BA learning within the timeout period and list the number of cases cannot be solved (\#Unsolved). The mark $n^*/m$ denotes that there are $n$ cases terminate with an error (in the over-approximation method) and it ran out of time for $m-n$ cases. The rows \#St., \#Tr., \#MQ, and \#EQ, are the numbers of states, transitions, membership queries, and equivalence queries. $\text{Time}_{eq}$ is the time spent in answering equivalence queries and $\text{Time}_{total}$ is the total execution time.}
  \label{tab:table-norm}
  \centering
  \scalebox{0.9}{
  \begin{tabular}{l|rr|rr|rr|rr|rr|rr|rr}
\toprule
Models & \multicolumn{2}{c|}{$\Ldollar$} & \multicolumn{4}{c|}{$\Lperiodic$} &
\multicolumn{4}{c|}{$\Lsyntactic$} & \multicolumn{4}{c}{$\Lrecurrent$}\\
\hline
\multirow{2}{1cm}{Struct.\& Approxi.} & \multirow{2}{*}{Table} & \multirow{2}{*}{Tree} & \multicolumn{2}{c|}{Table} & \multicolumn{2}{c|}{Tree} & \multicolumn{2}{c|}{Table} & \multicolumn{2}{c|}{Tree} & \multicolumn{2}{c|}{Table} & \multicolumn{2}{c}{Tree}\\
\cline{4-15}

 & & &  under& \multicolumn{1}{c|}{over} &  under& over&  under& \multicolumn{1}{c|}{over} &  under& over& under& \multicolumn{1}{c|}{over} &  under& over\\
\hline
\#Unsolved     & 4 & 2   & 3    & 0/2   & 2    & 0/1  & 1    & 4*/5 & 0    & 3*/3  & 1    & 0/1  & 1 & 0/1\\
\#St.       &3078 &3078  &2481  & 2468   &2526  &2417  &2591  &2591  &\textbf{2274}   &\textbf{2274}   &2382  &2382   &2400 &2400 \\
\#Tr.       &10.6k&\textbf{10.3k} &13.0k  &13.0k  &13.4k &12.8k &13.6k &13.6k &12.2k  &12.2k  &12.7k &12.7k  &12.8k &12.8k \\
\#MQ        &105k &114k  &86k   &85k     &69k   &\textbf{67k}   &236k  &238k  &139k   &139k   &124k  &124k   &126k &126k \\
\#EQ        &\textbf{1281} &2024  &1382  &1351    &1950  &1918  &1399  &1394  &2805   &2786   &1430  &1421   &3037 &3037 \\
$\text{Time}_{eq}$(s)      &146 &817    & 580  &92     &186   &159    &111   &115   &\textbf{89}     &91    &149   &149     &462 &465 \\
$\text{Time}_{total}$(s)      &183 &861    &610   &114    &213   &186    &140   &144   &\textbf{118}    &120    &175   &176    &499 &501 \\
EQ(\%)                     &79.8 &94.9  &95.1  &80.7   &87.3  &85.5   &79.3  &79.9  &\textbf{75.4}  &75.8   &85.1  &84.6    &92.6 &92.8 \\

\bottomrule
\end{tabular}
  }
\end{table}

The overall experimental results are given in Tab.~\ref{tab:table-norm}. In this section, we use $\Ldollar$ to denote the $\omega$-regular learning algorithm in~\cite{Farzan2008},
and $\Lperiodic$, $\Lsyntactic$, and $\Lrecurrent$ to represent the periodic,  syntactic, and recurrent $\FDFA$ learning algorithm introduced in Sec.~\ref{sec:fdfa-learner-table} and~\ref{sec:fdfa-learner-tree}. From the table, we can find the following facts:
(1) The BAs learned from $\Ldollar$ have more states but fewer transitions
than their $\FDFA$ based counterpart.
(2) $\Lperiodic$ uses fewer membership queries
comparing to $\Lsyntactic$ and $\Lrecurrent$. The reason is that $\Lsyntactic$ and $\Lrecurrent$ need to restart the learning of all progress automata from scratch when the leading automaton has been modified.
(3) Tree-based algorithms always solve more learning tasks than their table-based counterpart. In particular, the tree-based $\Lsyntactic$ with the under-approximation method solves all 295 learning tasks.

\begin{wrapfigure}{l}{5cm}
	\vspace*{-0.8cm}
	\centering
\pgfplotsset{compat=1.3}
\usetikzlibrary{calc,trees,shapes,intersections,arrows,automata,patterns,plotmarks}
\tikzset{astate/.style={draw,circle,inner sep=0pt,minimum size=8mm}}
\tikzset{astateq/.style={draw,rectangle,inner sep=3pt,rounded corners}}
\tikzset{adistr/.style={draw,circle,fill,minimum size=1mm,inner sep=0mm}}

\colorlet{colorcuddmtbdd}{brown}
\colorlet{colorcuddbdd}{blue}
\colorlet{colorbuddy}{green!60!black}
\colorlet{colorcacbdd}{violet}
\colorlet{colorjdd}{red}
\colorlet{colorbeedeedee}{cyan}
\colorlet{colorsylvan1}{teal}
\colorlet{colorsylvan7}{black}

\pgfplotsset{
	/pgfplots/ybar legend/.style={
		/pgfplots/legend image code/.code={%
		\draw[##1,/tikz/.cd,bar width=3pt,yshift=-0.3em,bar shift=0pt,draw=none]
			plot coordinates {(0mm,0.8em)};},
	},
	every axis y label/.style={at={(ticklabel cs:0.5)}, rotate=90,anchor=near ticklabel},
	every axis x label/.style={at={(ticklabel cs:0.5)}, anchor=near ticklabel},
	every axis legend/.append style={draw=none,column sep=4mm,semithick},
	every axis/.append style={clip =false},
}

\resizebox{150pt}{!}{
\begin{tikzpicture}
\scriptsize
\pgfplotsset{every axis legend/.append style={at={(axis description cs:0.5,-0.15)},anchor=north,legend columns=4,font=\large}}
\pgfplotsset{tick label style={font=\large}}
\pgfplotsset{title style={font=\Large}}
\pgfplotsset{label style={font=\large}}
	\node (cluster) at (0,0) {%
	\resizebox{40mm}{!}{
	\begin{tikzpicture}
	\begin{axis}[
		xlabel={Number of Step},
		ylabel={Number of States},
		width=100mm,
		minor x tick num=4,
		]
	\addplot+[mark=+,mark options={draw=none,thick,fill=blue},blue] coordinates { 
		(1, 1)
		(2, 5)
		(3, 9)
		(4, 16)
		(5, 17)
		(6, 17)
		(7, 25)
		(8, 36)
		(9, 65)
		(10, 109)
		(11, 156)
		(12, 177)
		(13, 187)
		(14, 202)
	};

    \addlegendentryexpanded{\kern-3mm $\Ldollar$}
	\addplot+[mark=+,mark options={draw=none,thick,fill=green!60!black},green!60!black] coordinates { 
		(1, 4)
		(2, 9)
		(3, 17)
		(4, 16)
		(5, 13)
		(6, 18)
		(7, 15)
		(8, 16)
		(9, 17)
		(10, 19)
		(11, 19)
		(12, 39)
	}; \addlegendentryexpanded{\kern-3mm $\Lperiodic$}
	\addplot+[mark=+,mark options={draw=none,thick,fill=violet},violet] coordinates { 
     (1, 4)
     (2, 10)
     (3, 17)
     (4, 16)
     (5, 13)
     (6, 18)
     (7, 24)
     (8, 31)
     (9, 37)
     (10, 37)
     (11, 44)
     (12, 44)
     (13, 49)
     (14, 57)
     (15, 57)
     (16, 65)
     (17, 86)
     (18, 86)
     (19, 86)
     (20, 86)
     (21, 89)
     (22, 19)
     (23, 19)
     (24, 19)
     (25, 22)
     (26, 28)
     (27, 35)
     (28, 38)
     (29, 38)
     (30, 48)
     (31, 48)
     (32, 48)
     (33, 60)
     (34, 60)
     (35, 60)
     (36, 63)
     (37, 26)
     (38, 26)
     (39, 29)
     (40, 29)
     (41, 31)
     (42, 38)
     (43, 39)
     (44, 45)
     (45, 45)
     (46, 48)
     (47, 43)
     (48, 57)
	}; \addlegendentryexpanded{\kern-3mm $\Lsyntactic$}
	\addplot+[mark=+,mark options={draw=none,thick,fill=red},red] coordinates { 
		(1, 4)
		(2, 9)
		(3, 12)
		(4, 11)
		(5, 9)
		(6, 9)
		(7, 10)
		(8, 10)
		(9, 11)
		(10, 11)
		(11, 15)
		(12, 17)
	}; \addlegendentryexpanded{\kern-3mm $\Lrecurrent$}
	\end{axis}
	\end{tikzpicture}
	}};

\end{tikzpicture}
}

	\caption{Growth of state counts in BA}
	\label{fig:growth-number-states-buchi-step}
	\vspace*{-0.6cm}
\end{wrapfigure}
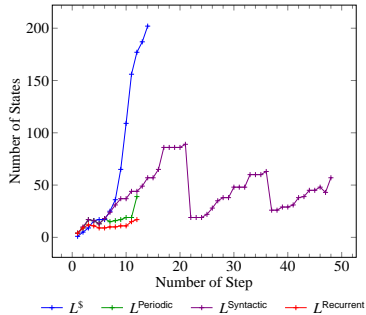

In the experiment, we observe that table-based $\Ldollar$ has 4 cases cannot be finished within the timeout period, which is the largest number amount all learning algorithms\footnote{Most of the unsolved tasks using the over-approximation method are caused by the situation that the $\FDFA$ teacher cannot find a valid counterexample for refinement. }.
We found that for these 4 cases, the average time required for $\Ldollar$ to get an equivalence query result is much longer than the $\FDFA$ algorithms. Under scrutiny, we found that the growth rate of the size (number of states) of the conjectured BAs generated by table-based $\Ldollar$ is much faster than that of table-based $\FDFA$ learning algorithms. In Fig.~\ref{fig:growth-number-states-buchi-step}, we illustrate the growth rate of the size (number of states) of the BAs generated by each table-based learning algorithm using one learning task that cannot be solved by $\Ldollar$ within the timeout period. The figures of the other three learning tasks show the same trend and hence are omitted.
Another interesting observation is that the sizes of BAs generated by $\Lsyntactic$ can decrease in some iteration because the leading automaton is refined and thus the algorithms have to redo the learning of all progress automata from scratch.

It is a bit surprise to us that, in our experiment, the size of BAs $\overline B$ produced by the over-approximation method is not much smaller than the BAs $\underline B$ produced by the under-approximation method.
Recall that the progress automata of $\overline B$ comes from the product of three DFAs $\machine^u_u\times (\proDFA^{u})^{s_u}_{v}\times (\proDFA^{u})^v_{v}$ while those for $\underline B$ comes from the product of only two DFAs $\machine^u_u\times (\proDFA^{u})^{s_u}_{v}$ (Sec.~\ref{sec:buechi-builder}).
We found the reason is that very often the language of the product of three DFAs is equivalent to the language of the product of two DFAs
, thus we get the same DFA after applying DFA minimizations.
Nevertheless, the over-approximation method is still helpful for $\Lperiodic$ and $\Lrecurrent$. For $\Lperiodic$, the over-approximation method solved more learning tasks than the under-approximation method. For $\Lrecurrent$, the over-approximation method solved one tough learning task that is not solved by the under-approximation method.

As we mentioned at the end of Sec.~\ref{sec:sub:fdfa-learner-tree-algo},
a possible optimization is to reuse the counterexample and to avoid equivalence query as much as possible.
The optimization helps the learning algorithms to solve nine more cases that were not solved before.

\section{Discussion and Future works}



Regarding our experiments, the BAs from LTL specifications are in general simple; the average sizes of the learned BAs are around 10 states. From our experience of applying DFA learning algorithms, the performance of tree-based algorithm is significantly better than the table-based one when the number of states of the learned DFA is large, say more than 1000. We believe this will also apply to the case of BA learning. Nevertheless, in our current experiments, most of the time is spent in answering equivalence queries. One possible direction to improve the scale of the experiment is to use a PAC (probably approximately correct) BA teacher~\cite{Angluin1988QCL} instead of an exact one, so the equivalence queries can be answered faster because the BA equivalence testing will be replaced with a bunch of BA membership testings.

There are several avenues for future works. We believe the algorithm and library of learning BAs should be an interesting tool for the community because it enables the possibility of many applications. For the next step, we will investigate the possibility of applying BA learning to the problem of reactive system synthesis, which is known to be a very difficult problem and learning-based approach has not been tried yet.

There are learning algorithms for residual NFA~\cite{Bollig2009ALN}, which is a more compact canonical representation of regular languages than DFA. We think maybe one can also generalize the learning algorithm for family of DFAs to family of residual NFAs (FRNFA). To do this, one needs to show FRNFAs also recognize $\omega$-regular language and finds the corresponding right congruences.

\bibliographystyle{abbrv}
\bibliography{notes}

\newpage
\appendix
\section*{Appendix}

In this section, we first show that although our acceptance condition defined in Sec.~\ref{sec:preliminaries}
is different from the original one defined
in \cite{Angluin2014}, but the ultimately periodic words of the $\FDFA$ will be preserved.
Then, we give the refinement for the progress trees of syntactic and recurrent $\FDFA$s in Sec.~\ref{app:ref_pt}.
In Sec.~\ref{app:approx}, we present the proofs of the lemmas given in Sec.~\ref{sec:buechi-builder}.
In Sec.~\ref{app:fa-construction}, we provide the constructions for the FAs $\autdollar_{u\$v}$, $\autdollarfdfa$ and $\autdollarfdfaneq$
as well as the correctness proof of counterexample analysis.
We also give the correctness proof and complexity of the tree-based learning algorithm in Sec.~\ref{app:correctness-and-termination-tree-algo}.

\section{Language Preservation under Different Acceptance Conditions}\label{app:lang-preservation}

Recall that the original acceptance condition for periodic $\FDFA$ in \cite{Angluin2014} is that $(u, v)$ is accepted by $\fdfas$ if
$v\in\lang{\proDFA^{\stateWord{u}}}$ where $\stateWord{u} = \machine(u)$.
While the original acceptance conditions for syntactic and recurrent $\FDFA$ in \cite{Angluin2014} are the same
as the one defined in this paper. More specifically, $(u, v)$ is accepted by $\fdfas$ if $\machine(uv)=\machine(u)$
and $v\in\lang{\proDFA^{\machine(u)}}$.
The set of ultimately periodic words of an $\FDFA$ $\fdfas$ is defined as $\upword{\fdfas} = \{uv^\omega \mid (u, v)\text{ is accepted by } \fdfas\}$.
The acceptance condition for periodic $\FDFA$ used in this paper is different from the original one in \cite{Angluin2014}.
We prove that the acceptance condition does not change the ultimately periodic words of the periodic $\FDFA$s.
\begin{lemma}\label{lem:ultimately-period-word-preservation}
Let $\fdfas$ be a periodic (syntactic, recurrent) $\FDFA$ under the acceptance condition in \cite{Angluin2014},
then $\upword{\fdfas}$ is preserved under the acceptance condition defined in this paper.
\end{lemma}
\begin{proof}
We only need to prove the preservation of ultimately periodic words for the periodic $\FDFA$s.
Given a periodic $\FDFA$ $\fdfas$,
the original acceptance condition of periodic $\FDFA$ requires that $(u, v)$ is accepted by $\fdfas$ if
$v\in\lang{\proDFA^{\stateWord{u}}}$ where $\stateWord{u} = \machine(u)$.
Clearly, the acceptance condition defined in this paper implies the original acceptance condition for the periodic $\FDFA$.
Therefore, we only need prove that if $(u, v)$ satisfies the original acceptance condition, then
there exists some decomposition $(x, y)$ of $\omega$-word $uv^\omega$ which satisfies our acceptance condition.
To achieve this, we first find a normalized formalization $(x, y)$ of $(u, v)$ such that $x = uv^i, y=v^j$ and $\eqWith{\machine}{xy}{x}$ for some $i\geq 0, j\geq 1$
according to \cite{Angluin2014}. Further, it is known that periodic $\FDFA$ is \emph{saturated} in the sense that
 under the original acceptance condition, if $(u, v)$
is accepted by $\fdfas$,
then every decomposition of $uv^\omega$ is accepted by $\fdfas$.
Therefore we have that $(x, y)$ is accepted by $\fdfas$, which means that $y\in\lang{\proDFA^{\stateWord{x}}}$
where $\stateWord{x} = \machine(x)$.
It follows that $(x, y)$ is accepted by $\fdfas$ under our acceptance condition.$\qed$
\end{proof}

We remark that in~\cite{Angluin2014}, they also define an acceptance condition called \emph{normalized acceptance condition},
which is able to make the syntactic and recurrent $\FDFA$s saturated in the sense that if $(u, v)$
is accepted by the $\FDFA$, then every decomposition of $uv^\omega$ is accepted by the $\FDFA$.
Since our goal is to learn a BA in this paper, we do not require the saturation property for all decompositions
of accepted $\omega$-word. Thus, we do not use the normalized acceptance condition.

\section{Refinement of the Progress Trees}\label{app:ref_pt}
Suppose $\stateWord{u} \cdot v^\omega \notin \upword{\omegaRegLang}$ for negative counterexample $(u, v)$,
we thus need refine the progress tree $\learnT_{\stateWord{u}}$.
Let $\size{v} = n$ and $h = s_0 s_1 \cdots s_n$ be the corresponding run of $v$ over $\proDFA^{\stateWord{u}}$.
At the beginning, we have $s_0 = \emptyword$ and $s_n = \stateWord{v}$ where $\stateWord{v} = \proDFA^{\stateWord{u}}(v)$
and $\stateWord{v}$ is an accepting state in $\proDFA^{\stateWord{u}}$, which implies that $\stateWord{u}(\stateWord{v})^\omega \in\upword{\omegaRegLang}$.
Our job here is to find the smallest $j\in [1\cdots n]$ such that
$\func{TE}(s_{j-1}, v[j\cdots n]) \neq \func{TE}(s_j, v[j+1\cdots n])$ so that we can use the experiment $e=v[j+1\cdots n]$ to differentiate $pa = s_{j-1}v[j]$ and $q=s_j$
since currently $s_j = \trans(s_{j-1}, v[j])$.

Afterwards, the progress tree $\learnT_{\stateWord{u}}$
can be refined by replacing the terminal node labeled with $s_j$ by a tree such that (i) its root is labeled by $e = v[j+1\cdots n]$,
(ii) its $\func{TE}(s_j, v[j+1\cdots n])$-subtree is a terminal node labeled by $s_j$,
and (iii) its $\func{TE}(s_{j-1}v[j], v[j+1\cdots n])$-subtree is a terminal node labeled by $s_{j-1}v[j]$.

In order to establish above result, we have to prove that $\func{TE}(s_0, v) \neq \func{TE}(s_n, \emptyword)$ to
ensure that there exists some $j\in[1\cdots n]$ such that $\func{TE}(s_{j-1}, v[j\cdots n]) \neq \func{TE}(s_j, v[j+1\cdots n])$. The proof is as follows.
\begin{itemize}
\item For periodic $\FDFA$, we have $\func{TE}(\emptyword, v)= \false$ since $\stateWord{u}(\emptyword\cdot v)^\omega\notin\upword{\omegaRegLang}$. Since $\stateWord{v}$ is an accepting state, we have $\func{TE}(\stateWord{v}, \emptyword) = \true$.
\item For syntactic $\FDFA$,
we notice that the counterexample requires $\eqWith{\machine}{uv}{u}$, that is,
$\stateWord{u} = \machine(uv) = \machine(u) = \machine(\stateWord{u}v)$.

First, we have $\func{TE}(\emptyword, v) = (\machine(\stateWord{u}\cdot \emptyword), \B) = (\stateWord{u}, \B)$,
where $\B$ is obtained here
since $\stateWord{u} = \machine(\stateWord{u}\cdot \emptyword\cdot v)$ and $\stateWord{u}(\emptyword\cdot v)^\omega\notin\upword{\omegaRegLang}$ according to the definition of $\func{TE}$ in syntactic $\FDFA$.

Since $\stateWord{v}$ is an accepting state in syntactic $\FDFA$, it follows that $\stateWord{u}=\machine(\stateWord{u}\stateWord{v})$ and $\stateWord{u}(\stateWord{v})^\omega\in\omegaRegLang$ according to Def.~\ref{def:cano-fdfas}.
Thus, we have $\func{TE}(\stateWord{v}, \emptyword) = (\machine(\stateWord{u}\stateWord{v}), \A) = (\stateWord{u}, \A)$
where $\A$ is obtained since $\stateWord{u}=\machine(\stateWord{u}\cdot\stateWord{v}\cdot\emptyword)$ and $\stateWord{u}(\stateWord{v}\cdot \emptyword)^\omega\in\upword{\omegaRegLang}$.

\item For recurrent $\FDFA$,
similar as in syntactic $\FDFA$, we have $\func{TE}(\emptyword, v) = \false$
and $\func{TE}(\stateWord{v}, \emptyword) = \true$.
\end{itemize}

We remark that, if the target is syntactic or recurrent $\FDFA$, as long as the leading automaton $\machine$ changes,
we need to initialize the classification tree $\learnT_u$ again for every state $u$ in leading automaton
since the labels on the edges depend on current leading automaton $\machine$.

\section{Proofs of Lem.~\ref{lem:fdfa-to-buechi-inclusion}, Lem~\ref{lem:word-fdfa-buechiL-preserve} and Lem~\ref{lem:lang-buechiU-period} }\label{app:approx}

\underword*
\begin{proof}
From the assumption, we have $\eqWith{\machine}{uv^k}{u}$
and $v^k \in \lang{\proDFA^{\stateWord{u}}}$ for any $k\geq 1$ where $\stateWord{u} = \machine(u)$.
It must be the case that some accepting state, say $f$ in $\proDFA^{\stateWord{u}}$, will
be visited twice after we read $v^n$ from initial state for some $n > \size{\proDFA^{\stateWord{u}}}$ with $f = \proDFA^{\stateWord{u}}(v^n)$
since $\proDFA^{\stateWord{u}}$ is a finite automaton.
In other words, there is a loop in the run of $v^n$ over $\proDFA^{\stateWord{u}}$.
Without loss of generality, suppose there exist $i, j\geq 1$ with $i + j = n$ such that
 $f = \proDFA^{\stateWord{u}}(v^i) = \proDFA^{\stateWord{u}}(v^{i+j})$.

In the following, our goal is to find some accepting state $f'$ such that $f'=\proDFA^{\stateWord{u}}(v^k) = \proDFA^{\stateWord{u}}(v^{2k})$
for some $k\geq 1$.
Fig.~\ref{fig:vk-progress} depicts how to find the accepting state $f'$ along the loop path in following two cases.
 \begin{itemize}
 \item $j\geq i$. Let $k=j$.

 \item $j< i$. Let $k = l\times j$ such that $k\geq i$ with the smallest $l\geq 1$.
 \end{itemize}
\begin{figure}
\centering
\begin{tikzpicture}[shorten >=1pt,node distance=1.5cm,on grid,auto,framed]

    \begin{scope}
       \node[initial,state, inner sep=3pt,minimum size=0pt] (q0)      {$s_{\stateWord{u}}$};
       \node[state, accepting, inner sep=3pt,minimum size=0pt]  (q1) [right =of q0]     {$f$};

       \node[state, accepting, inner sep=2.6pt,minimum size=0pt] (q2) [right =of q1] {$f'$};
       \node[] at ($(q0) + (-1, 0.7)$) {$j\geq i$};

       \path[->]
           (q0) edge node {$v^i$}   (q1)
           (q1) edge [bend left] node {$v^{j-i}$} (q2)
           (q2) edge [bend left] node {$v^{i}$} (q1)
           ;
    \end{scope}
    \begin{scope}[xshift=5.7cm]
           \node[initial,state, inner sep=3pt,minimum size=0pt] (q0)      {$s_{\stateWord{u}}$};
       \node[state, accepting, inner sep=3pt,minimum size=0pt]  (q1) [right =of q0]     {$f$};

       \node[state, accepting, inner sep=2.6pt,minimum size=0pt] (q2) [right =of q1] {$f'$};

       \node[] at ($(q0) + (-1, 0.7)$) {$j < i$};
       \path[->]
           (q0) edge node {$v^i$}   (q1)
           (q1) edge [bend left] node {$v^{c}$} (q2)
           (q2) edge [bend left] node {$v^{j-c}$} (q1)
           ;
    \end{scope}
\end{tikzpicture}
\caption{Finding $v^k$. If $j \geq i$, we let $k = j$, otherwise let $c = (l\cdot j-i)\% j \geq 0$ where $k = l\cdot j\geq i$ for some $l\geq 1$}\label{fig:vk-progress}
\end{figure}
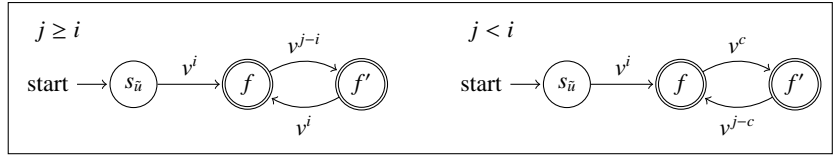

It is easy to check that $\proDFA^{\stateWord{u}}(v^k) = \proDFA^{\stateWord{u}}(v^{2k})$
since progress automaton $\proDFA^{\stateWord{u}}$ is deterministic
and the corresponding $f'$ is an accepting state.

It follows that
$v^k$ is accepted by the product $\underline{P}_{(\stateWord{u},f')}$ of three automata $\machine^{\stateWord{u}}_{\stateWord{u}}$
 , $(\proDFA^{\stateWord{u}})^{s_{\stateWord{u}}}_{f'}$ and $(\proDFA^{\stateWord{u}})^{f'}_{f'}$ where $s_{\stateWord{u}}$ is the initial state of $\proDFA^{\stateWord{u}}$.
In other words, $\omega$-word $uv^\omega$ will be accepted in $\buechiL$
since $u\cdot(v^k)^\omega \in \lang{\machine^{\initState}_{\stateWord{u}}}\cdot (\lang{\underline{P}_{(\stateWord{u},f')}})^\omega$.
$\qed$
\end{proof}

\overword*
\begin{proof}
Here we only consider ultimately periodic $\omega$-words in $\buechiU$, so every $\omega$-word can be given by a decomposition.

Since $\upword{\lang{\buechiU}} = \bigcup_{u \in \states, p\in\acc_u} \lang{\machine^{\initState}_{u}} \cdot (\lang{\overline{P}_{(u, p)}})^\omega$,
suppose $\omega$-word $w = uv^\omega \in \upword{\lang{\buechiU}}$, then $w$ can be given by a decomposition $(u, v)$
such that $u\in\lang{\machine^{\initState}_{\stateWord{u}}}$ and $v \in (\lang{\overline{P}_{(\stateWord{u}, p)}})^+$ for some $p \in \acc_{\stateWord{u}}$
where $\stateWord{u} = \machine(u)$.
Thus, we have $v = v_1\cdots v_n$ for some $n\geq 1$ such that $v_i\in\lang{\overline{P}_{(\stateWord{u}, p)}}$ for every $1\leq i\leq n$.
In addition, since $\overline{P}_{(\stateWord{u}, p)} = \machine^{\stateWord{u}}_{\stateWord{u}} \times (\proDFA^{\stateWord{u}})^{s_{\stateWord{u}}}_p$,
we conclude that $\eqWith{\machine}{uv}{u}$ and $v_i\in\lang{(\proDFA^{\stateWord{u}})^{s_{\stateWord{u}}}_p}$ for every  $1\leq i\leq n$
where $s_{\stateWord{u}}$ is the initial state in $\proDFA^{\stateWord{u}}$.

Observe that $p$ is the only accepting state of $(\proDFA^{\stateWord{u}})^{s_{\stateWord{u}}}_p$ and $(\proDFA^{\stateWord{u}})^{s_{\stateWord{u}}}_p$ is
obtained from $\proDFA^{\stateWord{u}}$ by setting $p\in\acc_{\stateWord{u}}$ as its only accepting state,
we have that $p = (\proDFA^{\stateWord{u}})^{s_{\stateWord{u}}}_p(v_i) = \proDFA^{\stateWord{u}}(v_i)$ for every $1\leq i\leq n$
and $p$ is an accepting state in $\proDFA^{\stateWord{u}}$.

The remaining job is how to find the accepting state $p$ in $\proDFA^{\stateWord{u}}$. Suppose we
have the counterexample $uv^\omega$ given by the decomposition $(u, v)$,
from which we construct the FA $\autdollar_{u\$v}$ by the method in Sec.~\ref{sec:ce-translation-dfa-omega-word}.
The number of states in $\autdollar_{u\$v}$ is in $\mathcal{O}(\size{v}(\size{v}+\size{u}))$.
In addition, we can construct an FA $\mathcal{A}$ such that
$\lang{\mathcal{A}} = \bigcup_{u\in\states, p\in\acc_u}\lang{\machine^{\initState}_u}\cdot \$ \cdot (\lang{\machine^u_u\times(\proDFA^u)^{s_u}_p})^+$
where $s_u$ is the initial state of $\proDFA^u$.
By fixing $u$ and $p$, we get $\lang{\mathcal{A}_{(u, p)}} = \lang{\machine^{\initState}_u}\cdot \$ \cdot (\lang{\machine^u_u\times(\proDFA^u)^{s_u}_p})^+ = \lang{\machine^{\initState}_u}\cdot \$ \cdot (\lang{\overline P_{(u, p)}})^+ $.
We get the corresponding $u$ and $p$ such that $\lang{\mathcal{A}_{(u, p)}\times \autdollar_{u\$v}} \neq \emptyset$.
There must exist such $u$ and $p$ otherwise $uv^\omega$ will not be accepted by $\buechiU$.
To get all the fragment words $v_i$ from $v$, one only needs to run the finite word $v$ over $\overline P_{(u, p)}$.
The time and space complexity of this procedure are in $\mathcal{O}(nk(n + nk)\cdot (\size{v}(\size{v}+\size{u})))$ and $\mathcal{O}((n + nk)\cdot (\size{v}(\size{v}+\size{u})))$ respectively
where $n$ is the number of states in the leading automaton and $k$ the number of states in the largest progress automaton.
Thus we complete the proof.$\qed$
\end{proof}

\bainclusion*
\begin{proof}
In the following, we prove the lemma by following cases.
\begin{itemize}
\item Sizes of $\buechiL$ and $\buechiU$.
In the under approximation construction, for every state $u$ in $\machine$, there is a progress automaton $\proDFA^u$ of size at most $k$.
It is easy to conclude that the automaton $\underline{P}_{(u, v)}$ is of size $nk^2$ for every $v\in\acc_u$, so
$\buechiL$ is of size $n+ nk\cdot nk^2 \in \mathcal{O}(n^2k^3)$.
The over-approximation method differs in the construction of the automaton $\overline{P}_{(u, v)}$ from
the under-approximation method.
It is easy to conclude that the automaton $\overline{P}_{(u, v)}$ is of size $nk$ for every $v\in\acc_u$, so
$\buechiU$ is of size $n+ nk\cdot nk \in \mathcal{O}(n^2k^2)$.

\item $\upword{\lang{\buechiL}}\subseteq \upword{\fdfas}$. Suppose ultimately periodic $\omega$-word $w$ is accepted by $\buechiL$, there must be an accepting run in $\buechiL$
in form of $\initState\xrightarrow{u}\stateWord{u}\xrightarrow{\emptyword} s_{\stateWord{u}, v}
\xrightarrow{v_1} f_v \xrightarrow{\emptyword} f'_v \xrightarrow{\emptyword} s_{\stateWord{u}, v} \cdots$.
Then the $\omega$-word $w$ can be divided into the form of $u\cdot \emptyword \cdot v_1 \cdot\emptyword\cdot \emptyword\cdot v_2 \cdots$
by $\epsilon$-transitions.
According to the construction of $\buechiL$, we have $u \in \lang{\machine^{\initState}_{\stateWord{u}}}$
and $v_i \in\lang{\overline{P}_{(\stateWord{u}, v)}}$ for any $i \geq 1$. Moreover, since $\overline{P}_{(\stateWord{u}, v)}$
is the product of three automata $\machine^{\stateWord{u}}_{\stateWord{u}}$, $(\proDFA^{\stateWord{u}})^{s_{\stateWord{u}}}_{v}$
and $(\proDFA^{\stateWord{u}})^{v}_{v}$ where $s_{\stateWord{u}}$ is the initial state in $\proDFA^{\stateWord{u}}$. It follows that
$\lang{\machine^{\initState}_{\stateWord{u}}}
\cdot (\lang{\overline{P}_{(\stateWord{u}, v)}})^*
= \lang{\machine^{\initState}_{\stateWord{u}}} $
and $(\lang{\overline{P}_{(\stateWord{u}, v)}})^+ = \lang{\overline{P}_{(\stateWord{u}, v)}}$.

By Lem.5 in\cite{Calbrix1993}, there exist two words $x \in \lang{\machine^{\initState}_{\stateWord{u}}}$
and $y \in \lang{\overline{P}_{(\stateWord{u}, v)}}$ such that $w = x\cdot y^\omega$.
In other words, we have $\stateWord{u} = \machine(x)$, $\eqWith{\machine}{xy}{x}$ and $y \in \lang{\proDFA^{\stateWord{u}}}$
, which implies that $w$ is accepted by $\fdfas$.

\item $\upword{\fdfas} \subseteq \upword{\lang{\buechiU}}$. Suppose an $\omega$-word $w \in\upword{\fdfas}$, then there exists a decomposition $(u, v)$ of $w$
such that $\eqWith{\machine}{uv}{u}$ and $\stateWord{v}$ is an accepting state where $\stateWord{u} = \machine(u)$
and $\stateWord{v} = \proDFA^{\stateWord{u}}(v)$. It follows that $v\in\lang{\overline{P}_{(\stateWord{u}, \stateWord{v})}}$ according
to Def.~\ref{def:fdfa-to-buechi}. In addition, we have $u \in \lang{\machine^{\initState}_{\stateWord{u}}}$,
which follows that $u \cdot v^\omega \in \lang{\machine^{\initState}_{\stateWord{u}}} \cdot (\lang{\overline{P}_{(\stateWord{u}, \stateWord{v})}})^\omega = \upword{\lang{\buechiU}}$.

\item $\upword{\lang{\buechiL}} = \upword{\fdfas}$ if $\fdfas$ is a canonical $\FDFA$. For any $\FDFA$ $\fdfas$,
we have $\upword{\lang{\buechiL}} \subseteq \upword{\fdfas}$. Thus, the remaining job is to prove that $\upword{\fdfas} \subseteq \upword{\lang{\buechiL}}$ if $\fdfas$ is a canonical $\FDFA$,
which follows from Prop.~\ref{prop:periodic-word-implication} and Lem.~\ref{lem:word-fdfa-buechiL-preserve}. Thus, we complete the proof.

\end{itemize}
$\qed$
\end{proof}

We present Prop.~\ref{prop:periodic-word-implication}, which follows from Def.~\ref{def:cano-fdfas} of the canonical $\FDFA$s.
\begin{proposition}\label{prop:periodic-word-implication}
Let $\omegaRegLang$ be an $\omega$-regular language,
$\fdfas=(\machine, \{\proDFA^{u}\})$ the corresponding periodic (syntactic, recurrent) $\FDFA$ and $u, v\in\finwords$.
We have that if $(u, v)$ is accepted by $\fdfas$ then
$(u, v^k)$ is also accepted by $\fdfas$ for any $k\geq 1$.
\end{proposition}
\begin{proof}
Let $\stateWord{u}=\machine(u)$ and $\stateWord{v^k} = \proDFA^{\stateWord{u}}(v^k)$
, then we have that $v^k \doubleEq^{{\tilde{u}}}_K\stateWord{v^k}$ for every $k\geq 1$
where $K\in\{P,S,R\}$.
This is because $\stateWord{v^k} = \proDFA^{\stateWord{u}}(\stateWord{v^k}) = \proDFA^{\stateWord{u}}(v^k)$
which makes $v^k$ in the equivalence class $\class{\stateWord{v^k}}$.
Our goal is to prove that $(u, v^k)$ is also accepted by $\fdfas$, that is,
$\eqWith{\machine}{uv^k}{u}$ and
$\stateWord{v^k}$ is an accepting state for every $k\geq 1$.
Since $\singleEq_\machine$ and $\canoEq$ is consistent in the three canonical $\FDFA$s,
so from the fact that $(u, v)$ is accepted by $\fdfas$, we have that $\eqWith{\machine}{uv}{u}$,
i.e., $\eqWith{L}{uv}{u}$. It follows that $\eqWith{L}{uv^k}{u}$ for every $k\geq 1$.
Thus, the remaining proof is to prove that $\stateWord{v^k}$ is an accepting state for every $k\geq 1$
in the three canonical $\FDFA$s.
\begin{itemize}
\item
For periodic $\FDFA$, since $(u, v)$ is accepted by $\fdfas$, i.e, $\stateWord{v}$ is an accepting state in $\proDFA^{\stateWord{u}}$,
then we have $\stateWord{u}(\stateWord{v})^\omega\in\omegaRegLang$ according to Def.~\ref{def:cano-fdfas}.
By definition of $\doubleEq^{\tilde{u}}_P$ and the fact that $\stateWord{v}\doubleEq^{\tilde{u}}_P v$,
we have that $\stateWord{u}(v)^\omega\in\omegaRegLang$, i.e., $\stateWord{u}(v^k)^\omega\in\omegaRegLang$ for every $k\geq 1$.
Similarly, since $\stateWord{u}(v^k)^\omega\in\omegaRegLang$ and $v^k \doubleEq^{\tilde{u}}_P \stateWord{v^k}$,
we conclude that $\stateWord{u}(\stateWord{v^k})^\omega\in\omegaRegLang$, which means that the state $\stateWord{v^k}$
is an accepting state in $\proDFA^{\stateWord{u}}$ for every $k\geq 1$.

\item
By the definition of $\doubleEq^{\tilde{u}}_R$, if $x \doubleEq^{\tilde{u}}_R y$, then
we have $\eqWith{L}{\stateWord{u}x}{\stateWord{u}}
\land \stateWord{u}x^\omega \in\omegaRegLang \Longleftrightarrow \eqWith{L}{\stateWord{u}y}{\stateWord{u}} \land \stateWord{u}y^\omega \in\omegaRegLang$
for any $x, y\in\finwords$.
Since $x\doubleEq^{\tilde{u}}_S y$ implies $x \doubleEq^{\tilde{u}}_R y$, we also have above result if $x \doubleEq^{\tilde{u}}_S y$.
In the following, $\doubleEq^{\tilde{u}}_K$ can be replaced by $\doubleEq^{\tilde{u}}_S$ and $\doubleEq^{\tilde{u}}_R$.

For syntactic $\FDFA$ and recurrent $\FDFA$, if $(u, v)$ is accepted
by $\fdfas$, then $\eqWith{L}{\stateWord{u}\stateWord{v}}{\stateWord{u}}$
and $\stateWord{u}(\stateWord{v})^\omega\in\omegaRegLang$ according to Def.~\ref{def:cano-fdfas}.
By the fact that $v \doubleEq^{\tilde{u}}_K \stateWord{v}$, if we set $x = v$ and $y = \stateWord{v}$,
then we have that $\eqWith{L}{\stateWord{u}v}{\stateWord{u}}$
and $\stateWord{u}(v)^\omega\in\omegaRegLang$, which implies that $\eqWith{L}{\stateWord{u}v^k}{\stateWord{u}}$
and $\stateWord{u}(v^k)^\omega\in\omegaRegLang$ for every $k\geq 1$.

Similarly, as $v^k \doubleEq^{\tilde{u}}_K \stateWord{v^k}$, if we set $x= v^k$ and $y = \stateWord{v^k}$, we have that
$\eqWith{L}{\stateWord{u}\stateWord{v^k}}{\stateWord{u}}$
and $\stateWord{u}(\stateWord{v^k})^\omega\in\omegaRegLang$, which follows that
$\stateWord{v^k}$ is an accepting state in $\proDFA^{\stateWord{u}}$ for every $k\geq 1$. $\qed$

\end{itemize}

\end{proof}

\section{Finite Automaton Construction and Correctness for Counterexample Analysis}\label{app:fa-construction}
\subsection{Construction for $\autdollar_{u\$v}$}\label{sec:ce-translation-dfa-omega-word}
In \cite{Calbrix1993}, they presented a canonical representation
$L_{\$} = \{u\$v \mid u\in\finwords, v\in\alphabet^+, uv^\omega\in\omegaRegLang\}$
for a regular $\omega$-language $\omegaRegLang$.
Theoretically, we can apply their method to obtain the $\autdollar_{u\$v}$ automaton for an $\omega$-word $uv^\omega$
where the number of states in $\autdollar_{u\$v}$ is in $\mathcal{O}(2^{\size{u}+\size{v}})$.
In this section, we introduce a more effective way
to build an automaton $\autdollar_{u\$v}$ such that
$L(\autdollar_{u\$v}) = \{u\$v \mid u\in\finwords, v\in\alphabet^+, uv^\omega = w\}$
for a given $\omega$-word $w$ with the number of states in $\mathcal{O}(\size{v}(\size{v}+\size{u}))$.
A similar construction for $\autdollar_{u\$v}$ has been proposed in \cite{Farzan2008}, which first computes the regular
expression to represent all possible decompositions of $uv^\omega$ and then constructs a DFA from the regular expression.
In this section, we give a direct construction for $\autdollar_{u\$v}$ of $uv^\omega$ as well as the complexity of the construction.

Fig.~\ref{fig:example-aut-single-omega-word} depicts an example automaton $\autdollar_{u\$v}$ for $\omega$-word $(ab)^\omega$.
From the example, we can find that both decompositions $(aba, ba)$ and $(ababa, bababa)$ have the same periodic
word $(ba)^\omega$, which means that the second finite word of a decomposition can be simplified as long as we do not change the periodic word.
\begin{figure}
\centering
\begin{tikzpicture}[shorten >=1pt,node distance=1.5cm,on grid,auto,framed]
      \node[initial,state] (q0)      {$q_0$};
      \node[state]         (q1) [right =of q0]  {$q_1$};
      \node[state]         (q2) [right =of q1]  {$q_2$};
      \node[state, accepting]         (q3) [right =of q2]  {$q_3$};

      \node[state] (q4) [below =of q0] {$q_4$};
      \node[state] (q5) [below =of q1] {$q_5$};
      \node[state] (q6) [below =of q2] {$q_6$};
      \node[state, accepting] (q7) [below =of q3] {$q_7$};

%
      \path[->] (q0)  edge node {$\$$} (q1)
                      edge [bend left] node {$a$}  (q4)
                (q1)  edge node {$a$} (q2)
                (q2)  edge [bend left ]node {$b$} (q3)
                (q3)  edge [bend left] node {$a$} (q2)

                (q4) edge [bend left] node {$b$} (q0)
                     edge node {$\$$} (q5)
                (q5) edge node {$b$} (q6)
                (q6) edge [bend left] node {$a$} (q7)
                (q7) edge [bend left] node {$b$} (q6);
\end{tikzpicture}
\caption{$\autdollar_{u\$v}$ for $\omega$-word $(aba, ba)$}\label{fig:example-aut-single-omega-word}
\end{figure}

Formally, we give the definition of a \emph{smallest period} in an $\omega$-word $w$ given by
its decomposition $(u, v)$ where $v\in\poswords$.
To that end, we need more notations.
We use $u\preOfeq v$ to represent that there exists some $j\geq 1$ such that $u=v[1\cdots j]$
, and we say $u$ is a prefix of $v$. We use $u\preOfneq v$ if $u\preOfeq v$ and $u\neq v$.
\begin{definition}[Smallest period]\label{def:smallest-period}
For any $\omega$-word $w$ given by $(u, v)$, we say $r$ is the smallest period of $(u, v)$
if $r \preOfeq v, r^\omega = v^\omega$ and for any $t \preOfneq r$, we have $t^\omega \neq r^\omega$.
\end{definition}
Take the $\omega$-word $(ab)^\omega$ as an example, $ab$ and $ba$ are the smallest periods
of decomposition $(ab, ab)$ and $(aba, ba)$ respectively.
It is interesting to see that $\size{ab} = \size{ba}$ and $ab$ can be transformed to $ba$ by
shift the first letter of $ab$ to its tail.
In general, given $\omega$-word $w$, the length of the smallest period is fixed no matter
how $w$ is decomposed which is justified by Lem.~\ref{lem:len-small-period}.

\begin{lemma}\label{lem:len-small-period}
Given an $\omega$-word $w$, $(u, v)$ and $(x, y)$ are different decompositions of $w$ and their corresponding
smallest periods are $r$ and $t$, respectively. Then $\size{r} = \size{t} = n$ and either there exists $j\geq 2$
such that $r = t[j\cdots n]\cdot t[1\cdots j-1]$ or $r = t$.
\end{lemma}
\begin{proof}
According to Def.~\ref{def:smallest-period}, $w = uv^\omega = ur^\omega = xy^\omega = xt^\omega$.
We prove it by contradiction. Without loss of generality, suppose $\size{r} > \size{t}$.
If $\size{u} = \size{x}$, then $r^\omega = t^\omega$, we then conclude that $r$ is not a
smallest period of $(u, v)$ since $t \preOfneq r$.
Otherwise if $\size{u} \neq \size{x}$, we can either prove that $r = t$ or
find some $j\geq 2$ such that $z = t[j\cdots n]\cdot t[1\cdots j-1] \preOfneq r$
and $z^\omega = r^\omega$ in following cases.
\begin{itemize}
\item $\size{u} > \size{x}$. Let $k = (\size{u} - \size{x})\% \size{t} + 1$. If $k = 1$, then $z = t$, otherwise $j = k$;
\item $\size{x} > \size{u}$. Let $k = (\size{r} - (\size{x} - \size{u})\% \size{r})\% \size{t} + 1$.
If $k = 1$, then $z = t$, otherwise $j = k$;
\end{itemize}
We depict the situation where $\size{u} > \size{x}$ in the following.
 \begin{figure}\label{fig:smallest-length}
  \centering
	\begin{tikzpicture}[->, >=stealth',shorten >=2pt,auto]
	
	\node (ur) at (-1, 0.5) {$(u, r)$};
	\node (s0) at ($(ur)+(4, 0)$) {$u[1]u[2]\cdots u[k]u[k+1] \cdots u[m] \cdot r\cdot r \cdot r \cdots $};

    \node (xt) at (-1, 0) {$(x, t)$};
	\node (q0) at ($(xt)+(4, 0)$) {$x[1]x[2]\cdots x[k]t[1] \cdots\cdot\cdot t[j-1] \cdot z\cdot z \cdot z \cdots $};
    	
	\end{tikzpicture}
\end{figure}

From the assumption $\size{t} < \size{r}$, we have that $z \preOfneq r$.
However, since $z^\omega = r^\omega$, we conclude that $r$ is not the smallest period of
$(u, v)$. Contradiction. Thus we complete the proof. $\qed$
\end{proof}

Lem.~\ref{lem:len-small-period} shows that if the size of the smallest period of an $\omega$-word $w$ is $n$,
then there are exactly $n$ different smallest periods for $w$.
In the following, we define the shortest form for a decomposition of an $\omega$-word.
\begin{lemma}\label{lem:decomposition-rewrite}
For any decomposition $(u, v)$ of an $\omega$-word $w$, and $y$ is its corresponding smallest
period, then we can rewrite $u = xy^i$ and $v=y^j$ for some $i\geq 0, j\geq 1$ such that
for any $x' \preOfeq u$ with $u=x'y^k$ for some $0\leq k\leq i$, we have $x' = xy^{i-k}$.
We say such $(x, y)$ is the \emph{shortest form} for $(u, v)$.
\end{lemma}
\begin{proof}
This can be proved by Def.~\ref{def:smallest-period} and the fact that $y^\omega = v^\omega$,
which can be further illustrated by the procedure of constructing $(x, y)$.
To find the shortest form of $(u, v)$, we need to first find the smallest period $y$ of $(u, v)$,
which is illustrated by following procedure. At first we initialize $k = 1$.
\begin{itemize}
\item Step 1. Let $y = v[1\cdots k]$, we recursively check whether there exists some $j\geq 1$ such that
$v = y^j$. If there exists such $j$, we return $y$ as the smallest period. Otherwise we go to Step 2.
\item Step 2. We increase $k$ by $1$ and go to Step 1.
\end{itemize}
Since $k$ starts at $1$, then $y$ must be the smallest period of $(u, v)$ such that $v^\omega = y^\omega$.

We find the above $x$ of the shortest form in the following procedure.
\begin{itemize}
\item Step 1. Let $x = u$. If $x=\emptyword$, or $x=y$ then we return $\emptyword$.
Otherwise we check whether there exists some $k\geq 1$ such that $x = x[1\cdots k]\cdot x[k+1\cdots \size{x}]$
and $y = x[k+1\cdots \size{x}]$.
If there is no such $k$, we return $x$ as the final result. Otherwise we go to Step 2.
\item Step 2. We set $u = x[1\cdots k]$.
\end{itemize}
One can easily conclude that $x$ is the shortest prefix of $u$ such that $u = xy^i$ for some $i\geq 0$. $\qed$
\end{proof}

Following corollary is straightforward.
\begin{corollary}\label{coro:shortest-form-decomposition}
Given two decompositions $(u_1, v_1)$ and $(u_2, v_2)$ of $uv^\omega$. If $(u_1, v_1)$ and $(u_2, v_2)$ share the smallest
period $y$, then they also have the same shortest form $(x, y)$ where $u_1 = xy^i, u_2=xy^j$ for some
$i,j\geq 0$.
\end{corollary}
\begin{proof}[Sketch]
If we assume they have different shortest forms, they should not be two decompositions of the same $\omega$-word. $\qed$
\end{proof}

By Coro.~\ref{coro:shortest-form-decomposition}, we can represent all decompositions of an $\omega$-word $w$ which share the same smallest period
$y$ with $(xy^i, y^j)$ with some $i\geq 0, j\geq 1$.
In addition, since the number of different smallest periods is $\size{y}$, we can thus denote all
the decompositions of $w$ by the set $\bigcup_{k=1}^{\size{y}}\{(x_k y_k^i, y_k^j) \mid i\geq 0, j\geq 1\}$ where $(x_k, y_k)$
is the $k$-th shortest form of $w$. Therefore, we provide the construction of $\autdollar_{u\$v}$ as follows.

\subsubsection{Construction of $\autdollar_{u\$v}$}
Now we are ready to give the construction of $\autdollar_{u\$v}$ for a single $\omega$-word $w$
given by $(u, v)$. Suppose $(x, y)$ is the shortest form of $(u, v)$, then we
construct $\autdollar_{u\$v}$ as follows. Let $k = 1$, $n = \size{y}$, and we first construct
an automaton $D_1$ such that $\lang{D_1} = xy^*\$y^+$.
\begin{itemize}
\item Step 1. If $k = n$, then we construct the $\autdollar_{u\$v}$ such that
$\lang{\autdollar_{u\$v}} = \bigcup^n_{i=1} \lang{D_i}$, otherwise, we go to Step 2.
\item Step 2. We first increase $k$ by $1$.
Let $u' = x\cdot y[1]$ and $y' = y[2\cdots n]\cdot y[1]$.
We then get the shortest form $(x', y')$ of $(u', y')$ where the second element is $y'$ since $y'$ is the smallest
period of $(u', y')$ according to Lem.~\ref{lem:len-small-period}.
We then construct an automaton $D_k$ such that $\lang{D_k} = x'y'^*\$y'^+$
and let $x= x', y = y'$ and go to Step 1.

\end{itemize}

Suppose $\size{x} = m$ and $\size{y} = n$, the DFA $\aut$ that accepts $xy^*\$y^+$ can be constructed as follows.
\begin{itemize}
\item If $m = 0$, then we construct a DFA $\aut = (\alphabet, \{q_0, \cdots, q_{2n}\}, q_0, \{q_{2n}\}, \trans)$
where we have that $\trans(q_{k-1}, y[k]) = q_k$ when $1\leq k\leq n-1$, $\trans(q_{n-1}, y[n]) = q_0$, $\trans(q_{0}, \$) = q_{n}$,
 $\trans(q_{n-1+k}, y[k]) = q_{n+k}$ when $1\leq k\leq n$, and $\trans(q_{2n}, y[1]) = q_{n+1}$.
\item Otherwise $m\geq 1$, then we construct a DFA $\aut = (\alphabet, \{q_0, \cdots, q_{2n+m}\}, q_0, \{q_{m+2n}\}, \trans)$ where
we have that $\trans(q_{k-1}, x[k]) = q_k$ when $1\leq k\leq m$, $\trans(q_{m-1+k}, y[k]) = q_{m+k}$ when $1\leq k\leq n-1$,
$\trans(q_{m+n-1}, y[n]) = q_{m}$, $\trans(q_{m}, \$) = q_{m+n}$, $\trans(q_{m+n+k-1}, y[k]) = q_{m+n+k}$ when $1\leq k\leq n$,
and $\trans(q_{m+2n}, y[1]) = q_{m+n+1}$.
\end{itemize}
One can validate that $\lang{\aut} = xy^*\$y^+$ and the number of states in $\aut$ is at most $\size{x} + 2\size{y} + 1$.

\begin{proposition}\label{prop:lang-omega-dfa}
Let $\autdollar_{u\$v}$ be the DFA constructed from the decomposition $(u, v)$ of $\omega$-word $uv^\omega$,
then $\lang{\autdollar_{u\$v}} = \{u'\$ v' \mid u'\in\finwords, v'\in\poswords, u'v'^\omega = uv^\omega\}$.
\end{proposition}
\begin{proof}
\item $\subseteq$.
This direction is easy since $\lang{\autdollar_{u\$v}}=\bigcup^n_{i=1}\lang{D_i}$, we only need to
prove that for any $1\leq i\leq n$, if $u'\$ v'\in \bigcup^n_{i=1}\lang{D_i}$, then $u'v'^\omega = uv^\omega$.
Suppose $\lang{D_i} = x_iy_i^*\$ y_i^+$, thus for any $u'\$v' \in\lang{D_i}$, we have $u' = x_iy_i^j$ and $v'=y_i^k$ for
some $j\geq 0, k\geq 1$. It follows that $u'v'^\omega = uv^\omega$ since $x_iy_i^\omega = uv^\omega$.
\item $\supseteq$.
For any decomposition $(u', v') $ of $uv^\omega$, we can get its shortest form $(x', y')$ where $y'$ is the smallest
period of $(u', v')$ according to Lem.~\ref{lem:decomposition-rewrite}.
Suppose $(x, y)$ is the first shortest form used in the $\autdollar_{u\$v}$ construction.
By Lem.~\ref{lem:len-small-period}, we prove $u'\$v'$ is accepted by $\autdollar_{u\$v}$ as follows.
\begin{itemize}
\item $y = y'$. We have that $u' = xy^i$ and $v'= y^j$ for some $i\geq 0, j\geq 1$, thus $u'\$v'\in\lang{D_1}\subseteq\lang{\autdollar_{u\$v}}$.
\item $y' = y[j\cdots n]y[1\cdots j-1]$ for some $j\geq 2$. We conclude that $\lang{D_j} = x'y'^*\$y'^+$
since the shortest form is unique if we fix the smallest period by Coro.~\ref{coro:shortest-form-decomposition},
which follows that $u'\$v'\in\lang{D_j}\subseteq \lang{\autdollar_{u\$v}}$.
\end{itemize}
Therefore, we complete the proof. $\qed$
\end{proof}

\begin{proposition}\label{prop:size-of-ddollar}
Given an $\omega$-word $w$ given by $(u, v)$, then the automaton
$\autdollar_{u\$v}$ has at most $\mathcal{O}(\size{v}(\size{u}+\size{v})$ of states.
\end{proposition}

For every automaton $D_i$ such that $\lang{D_i} = xy^*\$y^+$, the number of states in $D_i$ is at most $\size{u} + 2\size{r} + 2$
where $r$ is the smallest period of $(u, v)$,
thus the number of states in $\autdollar_{u\$v}$ is in $\mathcal{O}(\size{r}\times(\size{r} + \size{u}))\in \mathcal{O}(\size{v}(\size{u}+\size{v})$.

\subsection{Construction of $\autdollarfdfa$ and $\autdollarfdfaneq$}\label{sec:ce-translation-dfa-fdfa}
In this section, given an $\FDFA$ $\fdfas = (\machine, \{\proDFA^u\})$, we provide the constructions for $\autdollarfdfa$
 and $\autdollarfdfaneq$. To ease the construction, we define two automata $N_u$ and $\overline{N_u}$ which will be used in the construction for every state $u$ in the leading automaton $\machine$. Assume that we have $\machine^u_u$, the corresponding
 progress automaton $\proDFA^{u} = ( \alphabet, \states^u, s^u, \acc^u, \trans^{u})$ and
 a DFA $\overline{\proDFA^u} = ( \alphabet, \states^u, s^u, \states^u\setminus\acc^u, \trans^{u})$ built from $\proDFA^u$
 such that $\lang{\overline {\proDFA^u}} = \finwords \setminus\lang{\proDFA^u}$. Note that the transition $\trans^{u}$ is complete
 in the sense that $\trans^{u}(s, a)$ is defined for every $s\in\states^u, a\in\alphabet$.
 \begin{itemize}
 \item For $\autdollarfdfa$, we have $N_u = \machine^u_u \times \proDFA^u$. Intuitively, we only keep the finite words which start at $u$ and can go back to $u$ in the leading automaton. In other words, $\lang{N_u} = \{ v\in\finwords \mid \eqWith{\machine}{uv}{u}, v\in\lang{\proDFA^u}\}$.

 \item For $\autdollarfdfaneq$, we have $\overline{N_u} = \machine^u_u \times \overline{\proDFA^u}$.
 Similarly, we have $\lang{\overline{N_u}} = \{ v\in\finwords \mid \eqWith{\machine}{uv}{u}, v\notin\lang{\proDFA^u}\}$.
 \end{itemize}

More precisely, The construction is as follows.
\begin{definition}\label{def:fdfa-to-dfa}
Let $\fdfas =\{\machine, \{\proDFA^u\}\}$ be an $\FDFA$ where we have
$\machine = ( \alphabet, \states, \initState, \trans)$
and for every $u \in \states$, the corresponding progress automaton
$\proDFA^{u} = ( \alphabet, \states^u, s^u, \acc^u, \trans^{u})$.
Let $N_{u}$ (and $\overline N_u$) be given by
$( \alphabet, \states_{u}, s_{u}, \acc_{u}, \trans_{u})$.
The DFA $\autdollarfdfa$ (and $\autdollarfdfaneq$) is defined as the tuple
$( \alphabet\cup\{\$\}, \states\cup \states_{Acc}, \initState, \acc, \trans \cup \trans_{Acc} \cup \trans_{\$})$
where
\[
\states_{Acc} = \bigcup_{u\in \states} \states_{u} \text{   and   }
\acc = \bigcup_{u\in \states} \acc_{u}
 \text{   and   }
\trans_{Acc} = \bigcup_{u\in\states} \trans_{u}
\]
\[
\trans_{\$} = \{(u, \$, s_u) \mid u\in\states\}
\]
where $\$$ is a fresh symbol.
\end{definition}

In Fig.~\ref{fig:example-fdfa-to-dfa}, we depict the DFA $\autdollarfdfa$ and $\autdollarfdfaneq$ constructed from $\fdfas$ in Fig.~\ref{fig:fdfa-example}.
\begin{figure}
\centering
\begin{tikzpicture}[shorten >=1pt,node distance=1.5cm,on grid,auto,framed]

    \begin{scope}
       \node[initial,state, inner sep=3pt,minimum size=0pt] (q0)      {$q_0$};
       \node[state, inner sep=3pt,minimum size=0pt]  (q1) [right =of q0]     {$q_1$};

       \node[state, accepting, inner sep=3pt,minimum size=0pt] (q2) [right =of q1] {$q_2$};
       \node[] at ($(q0) + (-1, 0.7)$) {$\autdollarfdfa$};

       \path[->]
           (q0) edge [loop above] node {a} (q0)
                edge [loop below] node {b} (q0)
                edge node {$\$$}   (q1)
           (q1) edge node {a, b} (q2)
           (q2) edge [loop above] node {a} (q2)
                edge [bend left] node {b} (q1)
           ;
    \end{scope}
    \begin{scope}[xshift=5.7cm]
           \node[initial,state, inner sep=3pt,minimum size=0pt] (q0)      {$q_0$};
       \node[state, accepting, inner sep=3pt,minimum size=0pt]  (q1) [right =of q0]     {$q_1$};

       \node[state, inner sep=3pt,minimum size=0pt] (q2) [right =of q1] {$q_2$};

       \node[] at ($(q0) + (-1, 0.7)$) {$\autdollarfdfaneq$};
       \path[->]
           (q0) edge [loop above] node {a} (q0)
                edge [loop below] node {b} (q0)
                edge node {$\$$}   (q1)
           (q1) edge node {a,b} (q2)
           (q2) edge [loop above] node {a} (q2)
                edge [bend left] node {b} (q1)
           ;
    \end{scope}
\end{tikzpicture}
\caption{$\autdollarfdfa$ and $\autdollarfdfaneq$ for $\fdfas$ in Fig.~\ref{fig:fdfa-example}}\label{fig:example-fdfa-to-dfa}
\end{figure}

\begin{proposition}\label{thm:language-dollar-dfa-fdfa}
Given an $\FDFA$ $\fdfas = (\machine, \{\proDFA^u\})$ and $\autdollarfdfa$ defined in Def.~\ref{def:fdfa-to-dfa}, then
$\lang{\autdollarfdfa} = \{ u \$ v \mid u\in\finwords, v\in\finwords
, \eqWith{\machine}{uv}{u}, \stateWord{u} = \machine(u), v\in\lang{\proDFA^{\stateWord{u}}} \}$.
\end{proposition}
\begin{proof}
By Def.~\ref{def:fdfa-to-dfa}, it is easy to conclude that for any $u\in\finwords$,
then we have $\stateWord{u} = \machine(u) = \autdollarfdfa(u)$. For any $u, v\in\finwords$,
we have that $N_{\stateWord{u}}(v) = \autdollarfdfa(u\$v)$ where $\stateWord{u} = \machine(u)$ since $\autdollarfdfa$ is a DFA.
By acceptance condition, $(u, v)$ is accepted by $\fdfas$ iff we have $\eqWith{\machine}{uv}{u}$ and $v\in\lang{\proDFA^{\stateWord{u}}}$ where $\stateWord{u}= \machine(u)$. Thus we just need to prove that $(u, v)$ is accepted by $\fdfas$ iff $u\$v$ is accepted by $\autdollarfdfa$.

\item $\supseteq$. $(u, v)$ is accepted by $\fdfas$, then $u\$v \in \lang{\autdollarfdfa}$.
By $\eqWith{\machine}{uv}{u}$ and $v\in\lang{\proDFA^{\stateWord{u}}}$,
we have that $v \in \lang{N_{\stateWord{u}}}$, which follows that $N_{\stateWord{u}}(v)$ is an accepting state.
Since $N_{\stateWord{u}}(v) = \autdollarfdfa(u\$v)$, we have that $\autdollarfdfa(u\$v)$ is an accepting state.
Therefore, $u\$v \in\lang{\autdollarfdfa}$.
\item $\subseteq$. First, we have that $\lang{\autdollarfdfa}\subseteq \finwords\$\finwords$ by Def.~\ref{def:fdfa-to-dfa}.
For any $u, v\in\finwords$, if $u\$v \in \lang{\autdollarfdfa}$, then $\autdollarfdfa(u\$v)$ is an accepting state.
It follows that $v \in \lang{N_{\stateWord{u}}}$ with $\stateWord{u} = \machine(u)$.
Since $N_{\stateWord{u}} = \machine^{\stateWord{u}}_{\stateWord{u}}\times \proDFA^{\stateWord{u}}$, we have that
$v\in \lang{\machine^{\stateWord{u}}_{\stateWord{u}}}$ and $v\in\lang{\proDFA^{\stateWord{u}}}$, which implies
that $\eqWith{\machine}{uv}{u}$ and $v\in\lang{\proDFA^{\stateWord{u}}}$. Thus, we conclude that $(u, v)$ is accepted by
$\fdfas$. $\qed$
\end{proof}

\begin{proposition}\label{thm:language-neg-dollar-dfa-fdfa}
Given an $\FDFA$ $\fdfas$ and $\autdollarfdfaneq$ the corresponding deterministic automaton, then
$\lang{\autdollarfdfaneq} = \{ u \$ v \mid u\in\finwords, v\in\finwords
, \eqWith{\machine}{uv}{u}, \stateWord{u} = \machine(u), v\notin\lang{\proDFA^{\stateWord{u}}}\} $.
\end{proposition}
\begin{proof}
By Def.~\ref{def:fdfa-to-dfa}, it is easy to conclude that for any $u\in\finwords$,
then we have $\stateWord{u} = \machine(u) = \autdollarfdfaneq(u)$. For any $u, v\in\finwords$,
we have that $\overline{N}_{\stateWord{u}}(v) = \autdollarfdfaneq(u\$v)$ where $\stateWord{u} = \machine(u)$ since $\autdollarfdfaneq$ is a DFA.

\item $\supseteq$. Assume that we have $\eqWith{\machine}{uv}{u}$ and $v\notin\lang{\proDFA^{\stateWord{u}}}$ where $\stateWord{u}= \machine(u)$.
By $\eqWith{\machine}{uv}{u}$, we have that $v\in\lang{\machine^{\stateWord{u}}_{\stateWord{u}}}$.
Further, from $v\notin\lang{\proDFA^{\stateWord{u}}}$, we have that $v \in \lang{\overline{\proDFA^{\stateWord{u}}}}$.
It follows that $N_{\stateWord{u}}(v)$ is an accepting state.
Since $\overline{N}_{\stateWord{u}}(v) = \autdollarfdfaneq(u\$v)$, then $\autdollarfdfaneq(u\$v)$ is an accepting state.
Therefore, $u\$v \in \lang{\autdollarfdfaneq}$.

\item $\subseteq$. First, we have that $\lang{\autdollarfdfaneq}\subseteq \finwords\$\finwords$ by Def.~\ref{def:fdfa-to-dfa}.
For any $u, v\in\finwords$, if $u\$v \in \lang{\autdollarfdfaneq}$, then $\autdollarfdfaneq(u\$v)$ is an accepting state.
It follows that $v \in \lang{\overline{N}_{\stateWord{u}}}$ with $\stateWord{u} = \machine(u)$.
Since $\overline{N}_{\stateWord{u}} = \machine^{\stateWord{u}}_{\stateWord{u}}\times \overline{\proDFA^{\stateWord{u}}}$, we have that
$v\in \lang{\machine^{\stateWord{u}}_{\stateWord{u}}}$ and $v\in\lang{\overline{\proDFA^{\stateWord{u}}}}$, which implies
that $\eqWith{\machine}{uv}{u}$ and $v\notin\lang{\proDFA^{\stateWord{u}}}$. $\qed$
\end{proof}
\begin{proposition}
The numbers of states in $\autdollarfdfa$ and $\autdollarfdfaneq$ are both in $\mathcal{O}(n+n^2k)$.
\end{proposition}
Suppose $n$ is the number of states in $\machine$ and $k$ is the number of states in the largest progress automaton,
then the number of states in $\autdollarfdfa$ ($\autdollarfdfaneq$) is in $\mathcal{O}(n+n^2k)$.

\subsection{Correctness of Counterexample Analysis for $\FDFA$ Teacher}\label{sec:sub:ce-translation-correctness}
Given the counterexample $uv^\omega$ for the $\FDFA$ teacher, we prove the decomposition $(u', v')$ is
a counterexample for $\FDFA$ learner defined in Def.~\ref{def:ce-for-fdfa-learner} by following cases:

\begin{itemize}
\item $uv^\omega\in\upword{\omegaRegLang} \land uv^\omega\notin \upword{\fdfas}$.
By Def.~\ref{def:ce-for-fdfa-learner}, we know that $uv^\omega$ is a positive counterexample
and we return a counterexample $(u', v')$ such that $u'\$v' \in \lang{\autdollar_{u\$v}}\cap \lang{\autdollarfdfaneq}$.
We first need to prove that $\lang{\autdollar_{u\$v}}\cap \lang{\autdollarfdfaneq}$ is not empty.
Since $uv^\omega\notin \upword{\fdfas}$, then any decomposition of $uv^\omega$, say $(u, v)$,
is not accepted by $\fdfas$. Since $\machine$ is a DFA, we can always find a decomposition
$x = uv^i$ and $y = v^j$ from some $i\geq 0, j\geq 1$ such that $\eqWith{\machine}{xy}{x}$ according to \cite{Angluin2014}.
Therefore $(x, y)$ is also a decomposition of $uv^\omega$ and it is not accepted by $\fdfas$, that is,
$y \notin\lang{\proDFA^{\stateWord{x}}}$ where $\stateWord{x} = \machine(x)$. It follows that
$x\$y \in \lang{\autdollarfdfaneq}$ according to Thm.~\ref{thm:language-neg-dollar-dfa-fdfa}. Thus, we conclude that
$\lang{\autdollar_{u\$v}}\cap \lang{\autdollarfdfaneq}$ is not empty. We let $u'= x$ and $v'=y$, and it
is easy to validate that $(u', v')$ is a positive counterexample for $\FDFA$ learner. This case is applied
for case U1 and O1.
\item $uv^\omega\in\upword{\omegaRegLang} \land uv^\omega\in \upword{\fdfas}$.
In this case, $uv^\omega$ is a spurious positive counterexample,
which happens when we use the under-approximation method to construct the B\"uchi automaton.
Here we also return a counterexample $(u', v')$ such that $u'\$v' \in \lang{\autdollar_{u\$v}}\cap \lang{\autdollarfdfaneq}$.
Since $uv^\omega\in \upword{\fdfas}$, then there exists some decomposition of $uv^\omega$, say $(u, v)$,
is accepted by $\fdfas$. We observe that $uv^\omega\notin\upword{\lang{\buechiL}}$, which follows that
there exists some $k\geq 1$ such that $(u, v^k)$ is not accepted by $\fdfas$ by Lem.~\ref{lem:word-fdfa-buechiL-preserve}.
By $\eqWith{\machine}{uv}{u}$, we also have that $\eqWith{\machine}{uv^k}{u}$ since $\machine$ is a DFA.
It follows that $u\$v^k \in \lang{\autdollarfdfaneq}$. Therefore, we conclude that
$\lang{\autdollar_{u\$v}}\cap \lang{\autdollarfdfaneq}$ is not empty and for every finite word $u'\$v' \in \lang{\autdollar_{u\$v}}\cap \lang{\autdollarfdfaneq}$, we have $(u', v')$ is a positive counterexample for $\FDFA$ learner.
This case is applied for U3.
\item $uv^\omega\notin\upword{\omegaRegLang} \land uv^\omega\in \upword{\fdfas}$.
In this case, $uv^\omega$ is a negative counterexample, one has to return a counterexample
$(u', v')$ such that $u'\$v' \in \lang{\autdollar_{u\$v}}\cap \lang{\autdollarfdfa}$.
We first need to prove that $\lang{\autdollar_{u\$v}}\cap \lang{\autdollarfdfa}$ is not empty.
Since $uv^\omega\in \upword{\fdfas}$, then there exists some decomposition $(u', v')$ of $uv^\omega$ is accepted by
$\fdfas$. It follows that $u'\$v' \in\lang{\autdollarfdfa}$ by Thm.~\ref{thm:language-dollar-dfa-fdfa}.
Thus we conclude that $\lang{\autdollar_{u\$v}}\cap \lang{\autdollarfdfa}$ is not empty. Moreover,
it is easy to validate that $(u', v')$ is a negative counterexample for $\FDFA$ learner. This case is applied for U2 and O2.

\item $uv^\omega\notin\upword{\omegaRegLang} \land uv^\omega\notin \upword{\fdfas}$.
In this case, $uv^\omega$ is a spurious negative counterexample,
which happens when we use the over-approximation method to construct the B\"uchi automaton.
It is possible that we cannot find a valid decomposition $(u', v')$ to refine $\fdfas$.
According to the proof of Lem.~\ref{lem:lang-buechiU-period}, one can construct a decomposition $(u, v)$ of $uv^\omega$ and $n\geq 1$ such that
$v = v_1\cdot v_2\cdots v_n$ and for all $i \in [1\cdots n]$, $v_i\in\lang{\proDFA^{\machine(u)}}$ and $\eqWith{\machine}{uv_i}{u}$.
If we find some $i\geq 1$ such that $uv_i^\omega \notin\upword{\omegaRegLang}$, then we let $u' = u$ and $v'=v_i$.
Clearly, $(u', v')$ is a negative counterexample for $\FDFA$ learner. This case is applied for O3.
\end{itemize}

\section{Correctness and Termination of Tree-based Algorithm}\label{app:correctness-and-termination-tree-algo}

In the following, we need the notion called normalized
factorization introduced in \cite{Angluin2014}.
Recall that given a decomposition $(u, v)$ of $\omega$-word $uv^\omega$ and the leading automaton $\machine$, we can get
its normalized factorization $(x, y)$ with respect to $\machine$ such that $x = uv^i, y = v^j$ and $\machine(xy) = \machine(x)$
for some smallest $i\geq 0, j\geq 1$ since $\machine$ is finite.
\subsection{Correctness of Tree-based Algorithm for $\FDFA$}
Lem.~\ref{lem:eq-class-division} establishes the correctness of our tree-based algorithm for periodic $\FDFA$.
\begin{lemma}\label{lem:eq-class-division}
For the leading tree in all three $\FDFA$s and the progress trees in the periodic $\FDFA$,
the tree-based algorithm will never classify two finite words which are in the same equivalence class into two
different terminal nodes.
\end{lemma}
\begin{proof}
We prove by contradiction. Suppose there are two finite word $x_1, x_2\in\finwords$ which are in the same equivalence
class but they are currently classified into different terminal nodes in classification tree $\learnT$.
\begin{itemize}
\item $\learnT$ is the leading tree. We assume that $x_1 \canoEq x_2$.
Suppose $x_1$ and $x_2$ have been assigned to terminal nodes
$t_1$ and $t_2$ respectively with $t_1\neq t_2$. Therefore, we can find the least common ancestor $n$ from $\learnT$,
where $L_n(n)=(y, v)$ is supposed to be an experiment to differentiate $x_1$ and $x_2$.
Without loss of generality, we assume that $t_1$ and $t_2$ are in the left and right subtrees of $n$ respectively.
Therefore, we have $\func{TE}(x_1, (y, v)) = \false$ and $\func{TE}(x_2, (y, v))=\true$. It follows
that $x_1(yv)^\omega\notin\upword{\omegaRegLang}$ and $x_2(yv)^\omega\in\upword{\omegaRegLang}$, which implies that
$x_1 \not\canoEq x_2$. Contradiction.

\item $\learnT=\learnT_{u}$ is a progress tree in periodic $\FDFA$.
We assume that $x_1 \proEq_P x_2$.
Similarly, suppose $x_1$ and $x_2$ have been assigned to terminal nodes
$t_1$ and $t_2$ of $\learnT_{u}$ respectively with $t_1\neq t_2$.
Therefore, we can find the least common ancestor $n$ from $\learnT_{u}$,
where $L_n(n)= v$ is supposed to be an experiment to differentiate $x_1$ and $x_2$.
Without loss of generality, we assume that $t_1$ and $t_2$ are in the left and right subtrees of $n$ respectively.
Therefore, we have $\func{TE}(x_1, v) = \false$ and $\func{TE}(x_2, v) = \true$. It follows
that $u(x_1v)^\omega\notin\upword{\omegaRegLang}$ and $u(x_2v)^\omega\in\upword{\omegaRegLang}$, which implies that
$x_1 \not\proEq_P x_2$. Contradiction.
\end{itemize}
$\qed$
\end{proof}

Lem.~\ref{lem:eq-class-division} cannot apply to the progress trees in syntactic and recurrent $\FDFA$s
as the progress trees heavily rely on the current leading automaton.
In the following, we prove the correctness of syntactic and recurrent $\FDFA$.
We say the leading automaton $\machine$ is consistent with $\canoEq$ iff for any $x_1, x_2\in\finwords$,
we have $\machine(x_1) = \machine(x_2) \Longleftrightarrow x_1\canoEq x_2$.

\begin{lemma}\label{lem:eq-class-division-s-r}
For the progress trees in the syntactic and recurrent $\FDFA$,
the tree-based algorithm will never classify two finite words which are in the same equivalence class into two
different terminal nodes if the leading automaton $\machine$ is consistent with $\canoEq$.

If the tree-based algorithm classifies two finite words which are in the same equivalence class into two
different terminal nodes, then $\machine$ is not consistent with $\canoEq$ currently.
\end{lemma}

\begin{proof}
Intuitively, the progress trees $\learnT_{u}$ in syntactic and recurrent $\FDFA$s are constructed with respect to the current leading automaton.
We prove the lemma in following cases.
\begin{itemize}
\item $\learnT_{u}$ is a progress tree in syntactic $\FDFA$.
We assume that $x_1\proEq_S x_2$.
Suppose $x_1$ and $x_2$ have been assigned to terminal node
$t_1$ and $t_2$ of $\learnT_{u}$ respectively.
Therefore, we can find the least common ancestor $n$ from $\learnT_{u}$,
where $L_n(n)= v$ is supposed to be an experiment to differentiate $x_1$ and $x_2$.
Thus,  by the definition of $\func{TE}$ in syntactic $\FDFA$, we can assume that $d_1: =\func{TE}(x_1, v) = (\machine(ux_1), m_1)$ and $d_2:=\func{TE}(x_2, v) = (\machine(ux_2), m_2)$
where $m_1, m_2 \in \{A, B, C\}$.
Since $t_1$ and $t_2$ are in different subtrees of $n$,
we thus have $d_1 \neq d_2$,
that is, $\machine(ux_1) \neq \machine(ux_2)$ or $m_1\neq m_2$.

1) First we assume that $\machine$ is consistent with $\canoEq$.
\begin{itemize}
\item $\machine(ux_1) \neq \machine(ux_2)$.
Since $x_1\proEq_S x_2$, we have $ux_1\canoEq ux_2$, which implies that $\machine(ux_1) = \machine(ux_2)$. Contradiction.
\item $m_1\neq m_2$.
Since $x_1\proEq_S x_2$, we have $ux_1\canoEq ux_2$, which follows that $\machine(ux_1) = \machine(ux_2)$ since
$\machine$ is consistent with $\canoEq$. Moreover, we have that $\machine(ux_1v) = \machine(ux_2v)$ since
$\machine$ is deterministic. We discuss the values of $m_1$ and $m_2$ in the following.
\begin{itemize}
\item $ u = \machine(ux_1 v)$. It follows that $\eqWith{\omegaRegLang}{ux_1v}{u}$
since $\machine$ is consistent with $\canoEq$,
which implies that $u(x_1v)^\omega \in\upword{\omegaRegLang}\Longleftrightarrow u(x_2v)^\omega \in\upword{\omegaRegLang}$.
Moreover, we have $u = \machine(ux_2 v)$ since $\eqWith{\omegaRegLang}{ux_1}{ux_2}$.
Therefore, we conclude that
$m_1, m_2 \in \{A, B\}$ by the definition of $\func{TE}$. Without loss of generality, we let $m_1 = A$
and $m_2 = B$, which implies that $u(x_1v)^\omega \in\upword{\omegaRegLang}$ while $u(x_2v)^\omega \notin\upword{\omegaRegLang}$.
Contradiction.
\item $ u\neq \machine(ux_1 v)$. Thus, we have $m_1 = m_2 = C$, which follows that $d_1 = d_n$
since $\machine(ux_1) = \machine(ux_2)$. Contradiction.
\end{itemize}


\end{itemize}
Therefore, $t_1$ and $t_2$
cannot be different terminal nodes.

2) In this case, $\machine$ is not necessarily consistent with $\canoEq$.
\begin{itemize}
\item $\machine(ux_1) \neq \machine(ux_2)$. Let $s_1 = \machine(ux_1)$ and $s_2 = \machine(ux_2)$.
We have that $s_1$ and $s_2$ are classified into different terminal nodes in the leading tree $\learnT$
since $s_1\neq s_2$ and they are two labels of the terminal nodes. It follows
that $s_1 \not\canoEq s_2$ by Lem.~\ref{lem:eq-class-division}.
By $x_1\proEq_S x_2$, we have $ux_1\canoEq ux_2$, which implies that
$s_1\not\canoEq ux_1$ or $s_2\not\canoEq ux_2$, otherwise we get $s_1\canoEq s_2$.
Without loss of generality, suppose $s_1\not\canoEq ux_1$, then there exists some experiment $(y, v)$
to differentiate them. However, $ux_1$ is currently assigned into the equivalence class of $s_1$ since
$s_1 = \machine(ux_1)$. It follows that $\machine$ is not consistent with $\canoEq$.

\item $m_1\neq m_2$.

1) We assume that $ux_1 v\canoEq u$, then we have $ux_2 v\canoEq u$ since $ux_1\canoEq ux_2$ by $x_1\proEq_S x_2$,
which implies that $u(x_1v)^\omega \in\upword{\omegaRegLang} \Longleftrightarrow u(x_2v)^\omega \in\upword{\omegaRegLang}$.
If $\machine$ is consistent with $\canoEq$, we conclude that $m_1 = m_2 = A$ or $m_1 = m_2 = B$. Contradiction.
Therefore, $\machine$ is not consistent with $\canoEq$.

2) We assume that $ux_1 v\not\canoEq u$, then we can find some experiment $(y, z)$ to differentiate them.
It follows that $ux_2 v\not\canoEq u$ since $x_1\proEq_S x_2$ and $ux_1 \canoEq ux_2$.
Assume that $\machine$ is consistent with $\canoEq$, then we have that $ u \neq \machine(ux_1v)$ and
$ u \neq \machine(ux_2v)$, which implies that $m_1=m_2=C$. Contradiction.
Thus, $\machine$ is not consistent with $\canoEq$.
\end{itemize}

\item $\learnT_{u}$ is a progress tree in recurrent $\FDFA$. The analysis is similar
as the syntactic $\FDFA$. We assume that $x_1\proEq_R x_2$.
Suppose $x_1$ and $x_2$ have been assigned to terminal node
$t_1$ and $t_2$ of $\learnT_{u}$ respectively.
Therefore, we can find the least common ancestor $n$ from $\learnT_{u}$,
where $L_n(n)= v$ is supposed to be an experiment to differentiate $x_1$ and $x_2$.
Thus,  we can assume that $d_1: =\func{TE}(x_1, v)$ and $d_2:=\func{TE}(x_2, v)$ where $d_1, d_2\in\{\false, \true\}$.
Since $t_1$ and $t_2$ are in different subtrees of $n$,
we thus have $d_1 \neq d_2$.

1) We assume that $\machine$ is consistent with $\canoEq$.
Without loss of generality, suppose $d_1 = \false$ and $d_2 = \true$.
Since $d_2 = \true$, we have that $u = \machine(ux_2 v)$ and $u(x_2v)^\omega\in\upword{\omegaRegLang}$.
It follows that $ux_2v\canoEq u$ since $\machine$ is consistent with $\canoEq$. Moreover, we conclude
that $u = \machine(ux_1 v)$ and $u(x_1v)^\omega\in\upword{\omegaRegLang}$ by the fact that $x_1\proEq_R x_2$.
By the definition of $\func{TE}$, we have $d_1 = \true$. Contradiction. Therefore, $t_1$ and $t_2$
cannot be different terminal nodes.

2) $\machine$ is not necessarily consistent with $\canoEq$.
Without loss of generality, suppose $d_1 = \false$ and $d_2 = \true$.
Since $d_2 = \true$, we have that $u = \machine(ux_2 v)$ and $u(x_2v)^\omega\in\upword{\omegaRegLang}$.
Assume that $\machine$ is consistent with $\canoEq$,
it follows that $ux_2v\canoEq u$. Moreover, we conclude
that $u = \machine(ux_1 v)$ and $u(x_1v)^\omega\in\upword{\omegaRegLang}$ by the fact that $x_1\proEq_R x_2$.
By the definition of $\func{TE}$, we have $d_1 = \true$. Contradiction. Therefore, $\machine$ is not consistent with $\canoEq$.
\end{itemize}
$\qed$
\end{proof}

Once two finite words which are in the same equivalence class have been classified into two terminal nodes in the progress tree,
we can always prove that the leading automaton is not consistent with $\canoEq$. Therefore, the $\FDFA$ teacher
is able to return some counterexample to refine the leading automaton. If the leading automaton changes,
the $\FDFA$ learner should learn all progress automata from scratch with respect to current leading automaton.
At a certain point, the leading automaton $\machine$ will be consistent with $\canoEq$ since it will be
added a new state after every refinement. Thus, we conclude that the equivalence classes in the progress trees
will finally be classified correctly.

\begin{proposition}
Given the $\FDFA$ teacher that is able to answer membership and equivalence queries for $\FDFA$, the tree-based
$\FDFA$ learning algorithm can correctly classify all finite words.
\end{proposition}


\subsection{Complexity for Tree-based $\FDFA$ Learning Algorithm}

The counterexample guided refinement for $\fdfas$ shows that:
\begin{corollary}\label{coro:ce-state-increase}
Given a counterexample $(u, v)$ for $\FDFA$ learner,
the tree-based $\FDFA$ learner will either add a new state to the leading automaton $\machine$ or the corresponding progress
automaton $\proDFA^{\stateWord{u}}$.
\end{corollary}
Corollary.~\ref{coro:ce-state-increase} is a critical property for the termination of the tree-based $\FDFA$ learning algorithm
since each time we either make progress for the leading automaton or the corresponding progress automaton.

In Lem.~\ref{lem:eq-class-division-s-r}, we encounter a situation where the progress tree may classify
two finite words which are in the same equivalence class into two terminal nodes if
$\machine$ is not consistent with $\canoEq$. One may worry that if the $\FDFA$ teacher chooses to refine
the progress automaton continually, the refinement may not terminate. Lem.~\ref{lem:forc-current-leading}
shows that it will terminate since the number of equivalence classes of the progress automata with respect to $\machine$
is finite.
More precisely, if we fix the leading automaton $\machine$, we are actually learning a DFA induced by the right congruence
$x \proEq_{S'} y$ iff $\machine(ux) = \machine(uy)$ and for every $v\in\finwords$,
if $\machine(uxv) = u$, then $u(xv)^\omega \in\omegaRegLang \Longleftrightarrow u(xv)^\omega \in\omegaRegLang $.
One can easily verify that $x \proEq_{S'} y$ is a right congruence.
We remark that if $\machine$ is consistent with $\canoEq$, then $x \proEq_{S'} y$ is equivalent to $x \proEq_{S} y$.
\begin{lemma}\label{lem:forc-current-leading}
Given then leading automaton $\machine$, then for every
state $u$ in $\machine$, the index of $\proEq_{S'} $ is bounded by $\size{\states}\cdot \size{\proEq_P}$
where $\states$ is the state set of $\machine$.
\end{lemma}
\begin{proof}
We prove the lemma by giving the upper bound $\size{\states}\cdot \size{\proEq_P}$ of the index of $\proEq_{S'}$. We use $q_i$ to denote the state which
can be reached by $u$ for $1\leq i \leq n$ where $n$ is the number of states reachable by $u$. We classify any $x\in\finwords$ into
a equivalence class of $\proEq_{S'}$ as follows.

We first find $q_i = \machine(ux)$. Since for every $y \in\finwords$ with $q_i = \machine(uy)$,
we have $\machine(uxv) = \machine(uyv)$, thus those experiments $v\in\finwords$ with $\machine(uxv)\neq u$
are not able to differentiate $x$ and $y$. In other words, the value of $\machine(uxv) = u$ is not necessary here.
Therefore, if we only consider $x, y\in\finwords$ with $q_i=\machine(ux)=\machine(uy)$, the criterion to decide
whether $x$ and $y$ are in the same equivalence class is
to judge whether for any $v\in\finwords$, $u(xv)^\omega \in\omegaRegLang \Longleftrightarrow u(yv)^\omega \in\omegaRegLang$,
which is exactly the same definition for $\proEq_P$. Thus, we can find the notation $(q_i, \class{x}_{\proEq_P})$
to uniquely represent the equivalence class $\class{x}_{\proEq_{S'}}$.
Therefore, the index of the right congruence $\proEq_{S'} $
is $ n \cdot \size{\proEq_P} \leq \size{\states}\cdot \size{\proEq_P}$. $\qed$
\end{proof}

Similarly, if we fix the leading automaton $\machine$ and learn recurrent $\FDFA$,
we are actually learning DFA induced by the right congruence
$x \proEq_{R'} y$ iff for every $v\in\finwords$,
$\machine(uxv) = u \land u(xv)^\omega \in\omegaRegLang \Longleftrightarrow
\machine(uyv) = u \land u(yv)^\omega \in\omegaRegLang$. Since $x \proEq_{S'} y$
implies $x \proEq_{R'} y$, it follows that $\size{\proEq_{R'}}$
is smaller than $\size{\proEq_{S'}}$.

The implication from $\deqWithu{S'}{x}{y}$ to $\deqWithu{R'}{x}{y}$ can be easily established by
assuming $\deqWithu{S'}{x}{y}$ and then for any $v\in\finwords$, we have
that $\eqWith{\machine}{uyv}{u}\land u(yv)^\omega\in\omegaRegLang$ if $\eqWith{\machine}{uxv}{u}\land u(xv)^\omega\in\omegaRegLang$.
First, assuming that $\eqWith{\machine}{uxv}{u}\land u(xv)^\omega\in\omegaRegLang$ and $\deqWithu{S'}{x}{y}$, one can
easily conclude that $u(yv)^\omega\in\omegaRegLang$.
In addition, one can combine the result $\eqWith{\machine}{ux}{uy}$ from $\deqWithu{S'}{x}{y}$ and assumption $\eqWith{\machine}{uxv}{u}$
together to prove $\eqWith{\machine}{uyv}{u}$ since $\machine$ is deterministic and $\singleEq_{\machine}$ is an equivalence relation.

\begin{lemma}\label{lem:forc-current-recurrent-leading}
Given the leading automaton $\machine$, then for every
state $u$ in $\machine$, the index of $\proEq_{R'}$ is bounded by $\size{\states}\cdot\size{\proEq_P}$
where $\states$ is the state set of $\machine$.
\end{lemma}

Assume that $\fdfas= (\machine, \{\proDFA^u\})$ is the corresponding periodic $\FDFA$ recognizing $\omegaRegLang$.
Let $n$ be the number of states in $\machine$ of $\fdfas$ and $k$ be the number of states in the largest progress automaton of $\fdfas$.

\treecomplexity*
\begin{proof}
Thm.~\ref{thm:algo-tree-halt-correctness} can be concluded from Lem.~\ref{lem:eq-class-division},
Coro.~\ref{coro:ce-state-increase},
Lem.~\ref{lem:forc-current-leading},
and Lem.~\ref{lem:forc-current-recurrent-leading}.
Suppose $\fdfas = (\machine, \{\proDFA^{u}\}$ is the corresponding periodic $\FDFA$ recognizing $\omegaRegLang$.
The number of states in
$\machine$ is $n$ and $k$ is the number of the largest progress automaton in $\fdfas$.

Given a counterexample $(u, v)$, the number of membership queries is at most $\size{u}$
if we refine the leading automaton and is at most $\size{v}$ if we refine the progress automaton.
Therefore, the number of membership queries used in analyzing counterexample is bounded by $\size{u} + \size{v}$.
One can also use binary search to reduce the number of membership queries used by counterexample analysis
to $\log (\size{u} + \size{v})$.
Moreover, when the classification tree has been refined, we need to construct the corresponding $\machine$ or $\proDFA^{\machine(u)}$ again.
Suppose the new added terminal node is labeled by $p$, the terminal node which needs to refined is labelled by $q$
and the experiment is $e$.
We only need to compute the successors of $p$ and update the successors of the predecessors of $q$.
\begin{itemize}
\item Computing the successors of $p$ is to calculate $\trans(p, a)$ for every $a\in\alphabet$, which
requires $\size{\alphabet}\cdot h$ membership queries where $h$ is the height of the classification tree.
\item Updating the successors of the predecessors of $q$ is to calculate $\func{TE}(s, e)$ for every state label
$s$ and $a\in\alphabet$ such that currently we have $\trans(s, a) = q$, which requires at most $\size{\alphabet}\cdot m$
membership queries where $m$ is the number of states in current $\machine$ or $\proDFA^{\machine(u)}$.
\end{itemize}
Since the height of the classification tree is at most $m$, thus the number of membership queries needed for constructing
the conjectured DFA is at most $2\cdot m\cdot \size{\alphabet}$.
It follows that for the tree-based algorithm, the number of membership queries used in the counterexample guided refinement
is bounded by $\size{u} + \size{v} + 2m\cdot \size{\alphabet}$.
We remark that in the table-based algorithm, the number of membership queries used in the counterexample guided refinement
is bounded by $\size{u} + \size{v} + m + \size{\alphabet}\cdot m + \size{\alphabet}$.

We give the complexity of the tree-based algorithm as follows.
\begin{itemize}
\item For periodic $\FDFA$. During the learning procedure, when receiving a counterexample for $\FDFA$ learner,
the tree-based algorithm either adds a new state into the leading automaton or into the corresponding
progress automaton. Thus, the number of the equivalence queries is bounded by $n + nk$ since
the number of states in the target periodic $\FDFA$ is bounded by $n + nk$.
In periodic $\FDFA$, we have $m \leq n+k$ since every time we either refine the leading automaton or a progress automaton.
Therefore, the number of membership queries needed for the algorithm is bounded by
$(n + nk)\cdot (\size{u} + \size{v} + 2(n + k)\cdot \size{\alphabet})\in \mathcal{O}((n + nk)\cdot (\size{u} + \size{v} + (n+k)\cdot\size{\alphabet}))$ in the worst case.

\item For syntactic and recurrent $\FDFA$, when receiving a counterexample for $\FDFA$ learner, the tree-based algorithm
will first decide whether to refine the leading automaton or the progress automaton. If it decides to refine the leading automaton,
we need to initialize all progress trees with a single node labelled by $\emptyword$ again, so the number of states in the progress automata
of the $\FDFA$ may decrease at that point, otherwise it refines the progress automaton and the number of states in $\FDFA$ will increase by one.

In the worst case, the learner will try to learn the progress automata as much as possible. In other words, if current leading
automaton has $m$ states, the number of states in every progress automaton is at most $m \cdot k$ according to
Lem.~\ref{lem:forc-current-leading} and Lem.~\ref{lem:forc-current-recurrent-leading}.
When all progress trees cannot be refined any more, either the learning task finishes or the $\FDFA$ teacher returns
a counterexample to refine current leading automaton. For the latter case, the number of states in the leading automaton will increase
by one, that is, $m+1$, and we need to redo the learning work for all progress trees.
The number of states in all progress automata in the new $\FDFA$ is bounded by $(m+1)^2\cdot k$.
Therefore, the number of
equivalence queries needed for tree-based algorithm is bounded by $(1 + 1\cdot1\cdot k) + (1 + 2\cdot2\cdot k) + \cdots
(1 + (n-1)\cdot (n-1)\cdot k) + (1 + n\cdot n\cdot k) \in \mathcal{O}(n + n^3 k)$.
Similarly, in syntactic and recurrent $\FDFA$s, we have that $m \leq n + nk$ since the number of states in a progress automaton
is bounded by $nk$.
It follows that the number of membership queries needed for the algorithm is in
$\mathcal{O}((n + n^3k)\cdot (\size{u} + \size{v} + 2(n+nk)\cdot\size{\alphabet})) \in \mathcal{O}((n + n^3k)\cdot(\size{u} + \size{v}
+ (n+nk)\cdot \size{\alphabet}))$ in the worst case.
\end{itemize}

$\qed$
\end{proof}
\treespacecomplexity*
\begin{proof}
As we mentioned in Sec.~\ref{sec:fdfa-learner-table}, the $\FDFA$ learner can be viewed as a learner
consisting of many component DFA learners. For a component DFA learner, suppose the number of the states in the target DFA
is $m$, for table-based component DFA learner, the size of the observation table is in $\mathcal{O}((m + m\cdot\size{\alphabet}) \cdot m)$
since there are $m + m\cdot\size{\alphabet}$ rows and at most $m$ columns in the observation table in the worst case.
In contrast, for the tree-based
component DFA learner, the number of nodes in the classification tree is in $\mathcal{O}(m)$ since the number of terminal nodes
in the classification tree is $m$ and the number of internal nodes is at most $m-1$.
\begin{itemize}
\item For the periodic $\FDFA$, the number of states in the $\FDFA$ will increase after each refinement step.
Thus, it is easy to conclude that the space required for the leading automaton is in $\mathcal{O}(n)$
if we use tree-based learning algorithm and the space required by the table-based algorithm is in $\mathcal{O}((n + n\cdot\size{\alphabet}) \cdot n)$. Similarly, the space required by tree-based learning algorithm to learn each progress automaton is in $\mathcal{O}(k)$,
while for table-based algorithm, the space required is in $\mathcal{O}((k+k \cdot\ \size{\alphabet}) \cdot k)$.

\item For the syntactic and recurrent $\FDFA$. The learning procedure for the leading automaton is the same as periodic automaton.
Thus the space required by table-based and tree-based algorithm remain the same.

For learning progress automaton, the number of states in each progress automaton is at most $nk$ according to Lem.~\ref{lem:forc-current-leading} and Lem.~\ref{lem:forc-current-recurrent-leading}.
Therefore, for table-based algorithm, the space required is in $\mathcal{O}((nk+nk \cdot\ \size{\alphabet}) \cdot nk)$.
While for tree-based algorithm, the space required to learn each progress automaton is in  $\mathcal{O}(nk)$.

\end{itemize}

$\qed$
\end{proof}

\begin{restatable}{proposition}{teachercomplexity}\label{prop:algo-teacher-complexity}
	In $\FDFA$ teacher, suppose $n$ is the number of states in the leading automaton and $k$ is the number of states in
	the largest progress automaton in the input $\FDFA$ $\fdfas$ and the returned counterexample $uv^\omega$ has a decomposition
	$(u, v)$. Then
	\begin{itemize}
		\item the time and space complexity for building the BAs $\buechiL$ and $\buechiU$ are in
		$\mathcal{O}(n^2k^3)$ and $\mathcal{O}(n^2k^2)$ respectively, and
		\item for the under approximation method, the time and space complexity for analyzing the counterexample $uv^\omega$ are in
		$\mathcal{O}(n^2k\cdot (\size{v}(\size{v} + \size{u}))$,
        while for the over approximation method, the time and space complexity for analyzing the counterexample $uv^\omega$ are in
		$\mathcal{O}(n^2k^2\cdot (\size{v}(\size{v} + \size{u}))$ and in $\mathcal{O}(n^2k (\size{v}(\size{v} + \size{u}))$ respectively.
	\end{itemize}
\end{restatable}

\begin{proof}
Suppose the $\FDFA$ teacher currently needs to answer the equivalence query for $\FDFA$ $\fdfas = (\machine, \{\proDFA^u\})$.
Then the number of states in $\buechiL$ ($\buechiU$) is in $\mathcal{O}(n + n^2k^3)$ (respectively, $\mathcal{O}(n + n^2k^2)$).
In addition, the number of states in FA $\autdollarfdfa$ and $\autdollarfdfaneq$ are both in $\mathcal{O}(n + n^2k)$ and
the number of states in $\autdollar_{u\$v}$ is at most $\size{v}(\size{v} + \size{u})$ given that $(u, v)$ is a decomposition of the returned counterexample $uv^\omega$,
which can be applied to the under and the over approximation except for case O3 in the over approximation.
When we analyze the spurious negative counterexample, the time and space complexity are in $\mathcal{O}(nk(n+nk)\cdot (\size{v}(\size{v}+\size{u})))$ and $\mathcal{O}((n + nk)\cdot (\size{v}(\size{v}+\size{u})))$
according to Lem.~\ref{lem:lang-buechiU-period}. Therefore, we complete the proof.
$\qed$
\end{proof}

\batermination*
\begin{proof}
If we use the under-approximation method to construct the B\"uchi automaton, then
the BA learning algorithm will need to first learn a canonical $\FDFA$ to get the final B\"uchi automaton in the worst case.
This theorem is justified by Lem.~\ref{lem:language-fdfa-to-ba} and Lem.~\ref{lem:fdfa-to-buechi-inclusion}.
$\qed$
\end{proof}

\end{document}